\pdfoutput=1
\PassOptionsToPackage{table}{xcolor}
\documentclass[sigconf, nonacm]{acmart}
\usepackage{stmaryrd}
\setcitestyle{numbers,sort&compress}

\AtEndPreamble{%
\theoremstyle{acmdefinition}
\newtheorem{remark}[theorem]{Remark}}

\newcommand\vldbdoi{XX.XX/XXX.XX}

\newcommand\vldbavailabilityurl{}
\newcommand\vldbpagestyle{plain}

\usepackage{amssymb}
\usepackage{amsthm}
\newtheorem{theorem}{Theorem}[section]

\newtheorem{lemma}[theorem]{Lemma}
\theoremstyle{definition}
\newtheorem{example}[theorem]{Example}

\usepackage{adjustbox}
\newcommand{\Changed}[1]{{\color{black}#1}}

\newcommand{\Paragraph}[1]{\emph{#1}}

\usepackage[binary-units]{siunitx}
\DeclareSIUnit{\txn}{txn}
\DeclareSIUnit{\batch}{batch}
\graphicspath{{./DrawIO/}}

\newcommand{\Replica}[1][r]{\MakeUppercase{#1}}
\newcommand{\Client}[1][c]{\MakeLowercase{#1}}
\newcommand{\Replicas}[1][r]{\mathfrak{\MakeUppercase{#1}}}
\newcommand{\ID}[1]{\mathop{\textsf{id}}(#1)}
\newcommand{\T}{\tau}
\newcommand{\ma}{\mathbf{a}}
\newcommand{\m}{\mathbf{m}}
\newcommand{\n}{\mathbf{n}}
\newcommand{\f}{\mathbf{f}}

\newcommand{\Instance}[1]{\mathcal{I}_{#1}}
\newcommand{\Primary}[1]{\mathcal{P}_{#1}}

\newcommand{\rn}{\rho}

\newcommand{\Message}[2]{\textsc{#1}(#2)}

\newcommand{\Name}[1]{\textnormal{\textsc{#1}}}
\newcommand{\PBFT}{\Name{Pbft}}

\newcommand{\HS}{\Name{HotStuff}}

\newcommand{\Narwhal}{\Changed{\Name{Narwhal-HS}}}

\newcommand{\RCC}{\Name{RCC}}
\newcommand{\SpotLess}{\Name{SpotLess}}

\newcommand{\MirBFT}{\Name{MirBFT}}
\newcommand{\ISS}{\Name{ISS}}
\newcommand{\RDB}{\Name{Apache ResilientDB}}
\newcommand{\Bitcoin}{\Name{Bitcoin}}
\newcommand{\Ethereum}{\Name{Ethereum}}
\newcommand{\MinBFT}{\Name{MinBFT}}

\newcommand{\SH}[1]{\texttt{#1}}
\newcommand{\MAC}{\SH{MAC}}
\newcommand{\DS}{\SH{DS}}
\newcommand{\RDMS}{\SH{RDMS}}

\newcommand{\RVS}{\Name{RVS}}
\newcommand{\Chained}{\Name{Chained}}

\newcommand{\Cert}[2]{\llparenthesis#1\rrparenthesis_{#2}}

\newcommand{\TP}[1]{\mathit{T}_{#1}}

\newcommand{\abs}[1]{\lvert #1 \rvert}

\newcommand{\union}{\cup}
\newcommand{\intersect}{\cap}

\newcommand{\Good}{\cellcolor{green!20}}

\renewcommand{\v}[1][v]{#1}
\newcommand{\MName}[1]{\textsc{#1}}
\newcommand{\PP}{\mathbb{P}}
\newcommand{\PPproof}[1]{\operatorname{cert}(#1)}
\newcommand{\PPf}[1]{\operatorname{claim}(#1)}
\newcommand{\AcceptedSet}{\mathbb{CP}}

\newcommand{\Precedes}[1]{\operatorname{precedes}(#1)}
\newcommand{\Depth}[1]{\operatorname{depth}(#1)}
\newcommand{\RVar}[2]{\textit{#1}_{#2}}
\newcommand{\digest}[1]{\operatorname{digest}(#1)}

\newcommand{\timer}[1]{\textit{t}_{#1}}

\usepackage[noend]{algorithmic}
\usepackage{algorithm,float}


\usepackage{algorithmic}
\newcommand{\GETS}{:=}
\newenvironment{myprotocol}{
    \hrule
    \small
    \smallskip
    \algsetup{linenosize=\footnotesize}
    \begin{algorithmic}[1]
        
        \newcommand{\SPACE}{\item[]}
        \newcommand{\TITLE}[2]{\item[] \textbf{\underline{##1}} (##2) \textbf{:}\\[2pt]}
        \makeatletter
            \newcommand{\EVENT}[1]{\STATE \textbf{event} ##1 \textbf{do}\begin{ALC@g}}
            \newcommand{\ENDEVENT}{\end{ALC@g}}
        \makeatother
        
        \makeatletter
            \newcommand{\FUCTION}[2]{\STATE \textbf{function} \Name{##1}(##2) \textbf{do}\begin{ALC@g}}
            \newcommand{\ENDFUNCTION}{\end{ALC@g}}
        \makeatother
}{
    \end{algorithmic}%
    \hrule
}

\usepackage{tikz,pgfplots,pgfplotstable}

\usetikzlibrary{shapes,arrows,automata,positioning,cd}
\tikzset{
  spotlessedge/.style   = {black, ->, >=stealth},
}
\usepackage{xcolor}

\usetikzlibrary{arrows.meta,decorations.pathreplacing}
\tikzset{
    >=Stealth,
    smalltext/.append style={scale=0.7},
    dot/.style={circle,scale=0.35,draw=black,fill=black}
}

\usetikzlibrary{automata, positioning, arrows}
\usetikzlibrary{arrows.meta}
\tikzset{
    >=Stealth,
    plot/.append style={baseline,scale=0.475},
    label/.append style={font=\strut\footnotesize},
    dot/.style={circle,scale=0.35,draw=black,fill=black},
}
\pgfplotscreateplotcyclelist{mycyclelist}{
    black   ,every mark/.append style={solid,fill=\pgfplotsmarklistfill},mark=*\\ 
    red     ,every mark/.append style={solid,fill=\pgfplotsmarklistfill},mark=square*\\
    blue    ,every mark/.append style={solid,fill=\pgfplotsmarklistfill},mark=triangle*\\
    teal    ,every mark/.append style={solid,fill=\pgfplotsmarklistfill},mark=pentagon*\\
    brown   ,every mark/.append style={solid,fill=\pgfplotsmarklistfill},mark=halfsquare*\\ 
    orange  ,every mark/.append style={solid,fill=\pgfplotsmarklistfill},mark=halfcircle*\\
    violet  ,every mark/.append style={solid,fill=\pgfplotsmarklistfill,rotate=180},mark=halfdiamond*\\
}
\pgfplotscreateplotcyclelist{mycyclelistex}{
    black   ,every mark/.append style={solid,fill=\pgfplotsmarklistfill},mark=*\\ 
    red     ,every mark/.append style={solid,fill=\pgfplotsmarklistfill},mark=square*\\
}
\pgfplotsset{
    tick label style={font=\large},
    legend style={font=\Large,cells={anchor=west}},
    title style={font=\Large},
    label style={font=\Large},
    width=262.5pt,
    height=185pt,
    every axis/.append style={
        ylabel near ticks,
        xlabel near ticks,
        mark size=2.5pt,
        cycle list name=mycyclelist,
        font=\Large,
        y tick label style={
                        /pgf/number format/precision=1,
                        /pgf/number format/fixed,
                        /pgf/number format/fixed zerofill
                    }
    },
    barstyle/.append style={
        ybar,
        bar width={0.5cm},
        enlarge x limits=0.4,
        enlarge y limits={upper=0.025},
        ymin=0,
        xtick=data
    }
}

\pgfplotstableread{
    n rvs spotless
    4 16267 65069
    7 16008 112061
    10 15746 157467
    13 15482 201267
    16 15215 243448
    19 14947 284006
    22 14679 322942
    25 14410 360265
    28 14142 339415
    31 13875 314488
    34 13609 293019
    37 13345 274325
    40 13083 257895
    43 12823 243337
    46 12567 230346
    49 12313 218680
    52 12063 208144
    55 11816 198581
    58 11573 189862
    61 11334 181880
    64 11099 174543
    67 10868 167778
    70 10641 161518
    73 10419 155710
    76 10201 150306
    79 9987 145266
    82 9777 140553
    85 9572 136137
    88 9371 131990
    91 9175 128089
    94 8983 124413
    97 8795 120941
}\dataOneHunRQ

\pgfplotstableread{
    n rvs spotless
    4 31887 127549
    7 30985 216899
    10 30112 301126
    13 29267 380480
    16 28450 455203
    19 27659 525530
    22 26894 591687
    25 26155 569689
    28 25441 533543
    31 24750 501890
    34 24082 473902
    37 23437 448952
    40 22814 426554
    43 22211 406325
    46 21629 387957
    49 21066 371201
    52 20522 355849
    55 19996 341729
    58 19488 328698
    61 18996 316631
    64 18521 305426
    67 18062 294992
    70 17617 285251
    73 17187 276137
    76 16772 267590
    79 16369 259559
    82 15980 251997
    85 15603 244866
    88 15238 238128
    91 14885 231753
    94 14542 225711
    97 14211 219977
}\dataTwoHunRQ

\pgfplotstableread{
n	hs    rvs
16	0.6375378704071045  0.22204193115234375
22	1.7720103931427003  0.40560422658920287
28	9.096274342536926   0.5889748930931091
34	40.50678160429001   1.3100160193443298
}\viewSyncSimAve

\pgfplotstableread{
n	hs    rvs
16	6.148000001907349  3.9411609172821045
22	8.316951036453247   7.510558843612671
28	57.886991024017334   14.801375150680542
34	463.42599511146545   31.46478009223938
}\viewSyncSimWorst

\pgfplotstableread{
Nodes spotless hotstuff rcc pbft   oldnarwhal narwhal
4   120523   44312   219847 184428 268800   114156
16  287944   17789   319832 137976 234800	249348	
32  264767   12993   280143 92850  178400   210656
64  242908   7854    214816 56876  98000    141610
96  214501   6051    183488 43363  64800    111442
128 183818   4833    149894 34669  46400    77312
}\dataTputNodesFF

\pgfplotstableread{
Nodes spotless	rcc	hotstuff pbft narwhal
10	23398	18020	467 	3418  9984
50	106544	75673	2239	16943 44288
100	183213	147248	4314	34669 77312
200	234280	208268	8050	46420 120576
400	294002	256912	11512	47408 167040
}\dataTputBatch

\pgfplotstableread{
Failures spotlessa spotlessb spotlessc spotlessd
0 256485 242908 206501 183601	
1 248172 242250 202671  183042	
2 245082 239069 201269  185475	
3 239674 239026 202853  184054	
4 236903 237835 201610  183355	
6 213785 233234 200750  183450	
8 168065 230570 200981  182284				
10 121337 228088 202832 183461
}\dataTputCF
		
\pgfplotstableread{
Failures spotlessa spotlessb spotlessc spotlessd
0  256485   242908 206501 183601	
0.2 245082 237835 200750 182284	
0.4 236903 230570 198982 172252	
0.6 213785 210139 182746 161903	
0.8 168065 163222 152002 146027	
1  121337 111102 110854 107876
}\dataTputPF

\pgfplotstableread{
Failures spotless rcc hotstuff pbft narwhal
0 183601 149894	4833 34669 77312					
1 183042 103302 4523 31020 77235
2 185475 101930 4232 31323 77259
3 184054 100647 4021 31799 76930
4 183355 98643 3838 32032 77540
6 183450 96574 3634 32363 77020
8 182284 95073 3343 32626 77640
10 183461 93380 3228 32834 76700
}\dataTputCFMb
							
\pgfplotstableread{
Failures spotless rcc hotstuff pbft oldnarwhal narwhal
0   183601 149894 4833	34669	51600	77312
0.2 182284 95073 3343  32389    35600	77640
0.4 172252 92719 2290  33236    26800	76340
0.6 161903  91372 1408 34613    18400	74766
0.8 146027  89111 987  35735   13200	72756				
1   107876  87445 766  36023   10000    56115
}\dataTputPFMb

\pgfplotstableread{
instances spotless rcc
1 25510 90135
4 80523 242927
8 140659 279586
16 227383 281683
32 266353 280143
}\dataTputCI

\pgfplotstableread{
instances spotless rcc
1 9634 48595
8 61129 177180
16 107226 210985
32 181048 211728
64 242908 214816
}\dataTputCIa

\pgfplotstableread{
instances spotless rcc
1 5363 33035
16 65913 148254
32 118736 147638
64 167069 148444
128 183065 149474
}\dataTputCIb

\pgfplotstableread{
Tspotless Lspotless Ths Lhs Trcc Lrcc Tpbft Lpbft
139923	0.846212    11858	0.195513    155333	0.518037    48992	0.217712								
175921	0.894369    20254   0.242174    197401  0.572393    69989   0.254881
217152	0.93113     23114   0.335695    243079  0.577284    84556   0.294574	
244518	0.945054    23227   0.605909    282100  0.738614    97388   0.318758
257669	1.14536     23126   0.953592    282952  1.003145    97263   0.561341
255996  2.346805    23117   1.308435	280143	2.069686    96967   0.765159
}\dataTputLat

\pgfplotstableread{
Tspotless Lspotless Ths Lhs Trcc Lrcc Tpbft Lpbft Tnarwhal Lnarwhal
162363	0.935063   8167     0.211677    165862	1.108021    44783	0.254343    52200	0.265142	
192639	0.960736   9552     0.213122    183121  1.399315    55824   0.350255	67400	0.275286	
228215	0.991092   11182    0.215725    204454  1.656242    60656   0.474527	87400	0.459654	
233284	2.290712   14890    0.346846    210449  2.457307    61151   0.577017	97000	0.897161	
234144	3.600248   15113    2.29635     211844  3.619698	61145   2.201399	97000	2.188492	
235421  4.25789    15175    3.818992	211637	4.84626     60675   3.923234    96400	3.780828
}\dataTputLatb

\pgfplotstableread{
Tspotless Lspotless Ths Lhs Trcc Lrcc Tpbft Lpbft Tnarwhal Lnarwhal
129384	1.216809    4881	0.242168	98111	1.641854	25411	0.375534    49024   1.081713
162955	1.952068    6788    0.275639    120801  2.778643	31700   0.569062	76672   1.354285
170211	2.704908    8391    0.301805    133294  3.525055	32815   0.821208	76544   2.718381
175307	3.593024    9416    0.602067	138198  4.629623 	32218   2.132833	76672   5.383019
187966	6.456446    9463    2.574637	146379  8.545887	31243   4.742109    76672   10.201253
187734	14.010092   9323    6.281965    142659  17.408783    30942   8.253153    77056  17.212493
}\dataTputLatc

\pgfplotstableread{
Failures spotless0 spotless1 spotless2  spotless3   spotless4   rcc    rcc1
0  183601 183601 183601 183601 183601 149894 149894
1  183601 183042 183379	183782 184168 149894 103302											
2  183601 185475 182327 183782 184388 149894 101930
3  183601 184054 183149 182317 185909 149894 100647
4  183601 183355 182942 182018 185169 149894 98643
6  183601 183450 182342	181981 184516 149894 96574
8  183601 182284 182428 182149 183145 149894 95073
10 183601 183461 182841 182347 183713 149894 93380
}\dataTputBFa

\pgfplotstableread{
Failures spotless0 spotless1 spotless2 spotless3 spotless4 rcc rcc1
0  183601 183601  183601  183601 183601 149894  149894 
0.2  183601 182284 182428 182149 183145 149894  95073 
0.4 183601 172252  179427  179328 185125 149894  92719 
0.6 183601 161903  178193  178764 184659 149894  91372 				
0.8 183601 146027  177491  177600 185632 149894  89111
1  183601 107876  177319  177629 185243 149894  87445
}\dataTputBFb

\pgfplotstableread{
size spotless rcc pbft hotstuff narwhal
48  183782 145367 31705	4833 77312									
200 180380 141577 29322 4552 77216
400 178727 133860 16756 4356 77184
800 172572 129810 6949 3973 76416
1600 158460 124720 3834 2976 73472
}\dataTputTxnSize

\newcommand{\resultgraph}[5]{\begin{tikzpicture}[plot]
    \begin{axis}[xlabel={#3},ylabel={#4},title={#2},xtick={#5},legend to name={mainlegend},legend columns=-1, ymin={-4000}]
        \addplot table[x={Nodes},y={spotless}] {#1};
        \addplot table[x={Nodes},y={hotstuff}] {#1};
        \addplot table[x={Nodes},y={rcc}] {#1};
        \addplot table[x={Nodes},y={pbft}] {#1};
        \addplot table[x={Nodes},y={narwhal}] {#1};
        \legend{\SpotLess{}, \HS{},\RCC{}, \PBFT{}, \Narwhal{}}
    \end{axis}
\end{tikzpicture}}

\newcommand{\resultgraphresource}[6]{\begin{tikzpicture}[plot]
    \begin{axis}[xlabel={#3},ylabel={#4},title={#2},xtick={#5},legend to name={mainlegend},legend columns=-1, ymin={-4000},#6]
        \addplot table[x={resource},y={spotless}] {#1};
        \addplot table[x={resource},y={hotstuff}] {#1};
        \addplot table[x={resource},y={rcc}] {#1};
        \addplot table[x={resource},y={pbft}] {#1};
        \addplot table[x={resource},y={narwhal}] {#1};
        \legend{\SpotLess{}, \HS{},\RCC{}, \PBFT{}, \Narwhal{}}
    \end{axis}
\end{tikzpicture}}

\newcommand{\resultgraphd}[5]{\begin{tikzpicture}[plot]
    \begin{axis}[xlabel={#3},ylabel={#4},title={#2},xtick={#5},legend to name={mainlegend},legend columns=-1, ymin={-4000}]
        \addplot table[x={Nodes},y={spotless}] {#1};
        \addplot table[x={Nodes},y={hotstuff}] {#1};
        \addplot table[x={Nodes},y={rcc}] {#1};
        \addplot table[x={Nodes},y={pbft}] {#1};
        \addplot table[x={Nodes},y={narwhal}] {#1};
        \legend{\SpotLess{}, \HS{},\RCC{}, \PBFT{}, \Narwhal{}}
    \end{axis}
\end{tikzpicture}}

\newcommand{\resultgraphfailure}[5]{\begin{tikzpicture}[plot]
    \begin{axis}[xlabel={#3},ylabel={#4},title={#2},xtick={#5}, ymin={-4000}, legend to name={mainlegend2},legend columns=-1]
        \addplot table[x={Failures},y={spotlessa}] {#1};
        \addplot table[x={Failures},y={spotlessb}] {#1};
        \addplot table[x={Failures},y={spotlessc}] {#1};
        \addplot table[x={Failures},y={spotlessd}] {#1};
        \legend{\SpotLess{32}, \SpotLess{64}, \SpotLess{96}, \SpotLess{128}}
    \end{axis}
\end{tikzpicture}}

\newcommand{\resultgraphfailureb}[5]{\begin{tikzpicture}[plot]
    \begin{axis}[xlabel={#3},ylabel={#4},title={#2},xtick={#5}, ymin={-4000}, legend to name={mainlegend},legend columns=-1]
        \addplot table[x={Failures},y={spotless}] {#1};
        \addplot table[x={Failures},y={hotstuff}] {#1};
        \addplot table[x={Failures},y={rcc}] {#1};
        \addplot table[x={Failures},y={pbft}] {#1};
        \addplot table[x={Failures},y={narwhal}] {#1};
        \legend{\SpotLess{}, \HS{}, \RCC{}, \PBFT{}, \Narwhal{}}
    \end{axis}
\end{tikzpicture}}

\newcommand{\resultgraphfailurec}[5]{\begin{tikzpicture}[plot]
    \begin{axis}[xlabel={#3},ylabel={#4},title={#2},xtick={#5}, ymin={-4000}, legend to name={mainlegend},legend columns=-1]
        \addplot table[x={Failures},y={spotless}] {#1};
        \addplot table[x={Failures},y={hotstuff}] {#1};
        \addplot table[x={Failures},y={rcc}] {#1};
        \addplot table[x={Failures},y={pbft}] {#1};
        \addplot table[x={Failures},y={narwhal}] {#1};
        \legend{\SpotLess{}, \HS{}, \RCC{}, \PBFT{}, \Narwhal{}}
    \end{axis}
\end{tikzpicture}}

\newcommand{\resultgraphfailurek}[5]{\begin{tikzpicture}[plot]
    \begin{axis}[xlabel={#3},ylabel={#4},title={#2},xtick={#5}, ymin={-4000}, legend to name={byzantinelengend},legend columns=-1]
        \addplot table[x={Failures},y={spotless0}] {#1};
        \addplot table[x={Failures},y={spotless1}] {#1};
        \addplot table[x={Failures},y={spotless2}] {#1};
        \addplot table[x={Failures},y={spotless3}] {#1};
        \addplot table[x={Failures},y={spotless4}] {#1};
        \addplot table[x={Failures},y={rcc}] {#1};
        \addplot table[x={Failures},y={rcc1}] {#1};
        \legend{\SpotLess{}, $\Name{SPL}_{A1}$, $\Name{SPL}_{A2}$, $\Name{SPL}_{A3}$, $\Name{SPL}_{A4}$, \RCC{}, $\RCC{}_{A1}$}
    \end{axis}
\end{tikzpicture}}

\newcommand{\resultgraphfailurel}[5]{\begin{tikzpicture}[plot]
    \begin{axis}[xlabel={#3},ylabel={#4},title={#2},xtick={#5}, ymin={-4000}, legend to name={byzantinelengend},legend columns=-1]
        \addplot table[x={Failures},y={spotless0}] {#1};
        \addplot table[x={Failures},y={spotless1}] {#1};
        \addplot table[x={Failures},y={spotless2}] {#1};
        \addplot table[x={Failures},y={spotless3}] {#1};
        \addplot table[x={Failures},y={spotless4}] {#1};
        \addplot table[x={Failures},y={rcc}] {#1};
        \addplot table[x={Failures},y={rcc1}] {#1};
        \legend{\SpotLess{}, $\Name{SPL}_{A1}$, $\Name{SPL}_{A2}$, $\Name{SPL}_{A3}$, $\Name{SPL}_{A4}$, \RCC{}, $\RCC{}_{A1}$}
    \end{axis}
\end{tikzpicture}}

\newcommand{\resultgraphconcurrent}[5]{\begin{tikzpicture}[plot]
    \begin{axis}[xlabel={#3},ylabel={#4},title={#2},xtick={#5}, ymin={-4000}, ymax={300000}, legend to name={mainlegend3},legend columns=-1]
        \addplot table[x={instances},y={spotless}] {#1};
        \addplot table[x={instances},y={rcc}] {#1};
        \legend{\SpotLess{}, \RCC{}}
    \end{axis}
\end{tikzpicture}}

\newcommand{\resultgraphtputlat}[5]{\begin{tikzpicture}[plot]
    \begin{axis}[xlabel={#3},ylabel={#4},title={#2},xtick={#5}, ymin={0}, legend to name={mainlegend4},legend columns=-1,
		grid=major, ymax={9}, xmax={250000}, xmin={0}]
        \addplot table[x={Tspotless},y={Lspotless}] {#1};
        \addplot table[x={Ths},y={Lhs}] {#1};
        \addplot table[x={Trcc},y={Lrcc}] {#1};
        \addplot table[x={Tpbft},y={Lpbft}] {#1};
        \addplot table[x={Tnarwhal},y={Lnarwhal}] {#1};
        \legend{\SpotLess{}, \HS{},\RCC{}, \PBFT{}}
    \end{axis}
\end{tikzpicture}}

\newcommand{\resultgraphtxnsize}[5]{\begin{tikzpicture}[plot]
    \begin{axis}[xlabel={#3},ylabel={#4},title={#2},xtick={#5}, ymin={0}, legend to name={mainlegend4},legend columns=-1]
        \addplot table[x={size},y={spotless}] {#1};
        \addplot table[x={size},y={hotstuff}] {#1};
        \addplot table[x={size},y={rcc}] {#1};
        \addplot table[x={size},y={pbft}] {#1};
        \addplot table[x={size},y={narwhal}] {#1};
        \legend{\SpotLess{}, \HS{},\RCC{}, \PBFT{}, \Narwhal{}}
    \end{axis}
\end{tikzpicture}}

\newcommand{\graphSpotLess}{
    \definecolor{bar-green}{RGB}{105,200,76}
    \definecolor{bar-yellow}{RGB}{250,183,51}
    \definecolor{bar-orange}{RGB}{255,142,21}
    \definecolor{bar-red}{RGB}{255,13,13}
    
    \begin{tikzpicture}[plot, scale=1]
        \begin{axis}[
            scale=0.75,
                title style={font=\LARGE},
                label style={font=\LARGE},
                ybar,
                bar shift=2.5pt,
                bar width=5pt,
                scaled y ticks = false,
                xlabel={Time(s)},
                ylabel={Throughput(txn/s)},
                ytick={0, 50000, 100000, 150000, 200000, 250000, 300000},
                yticklabels={0, 50K,100K,150K,200K,250K,300K},
                xtick={0, 20, 40, 60, 80, 100, 120, 140},
                title={\SpotLess{} with $1$ failure},
                ymax=280000,
                ymin=0,
            ]
            \addplot[fill=bar-green] coordinates    {(0,   185013)};
            \addplot[fill=bar-green] coordinates    {(5,   189374)};
            \addplot[fill=bar-green] coordinates    {(10,   172567)};
            \addplot[fill=bar-green] coordinates    {(15,   185474)};
            \addplot[fill=bar-green] coordinates    {(20,   185398)};
            \addplot[fill=bar-green] coordinates    {(25,   198540)};
            \addplot[fill=bar-green] coordinates    {(30,   193025)};
            \addplot[fill=bar-green] coordinates    {(35,   191249)};
            \addplot[fill=bar-green] coordinates    {(40,   190152)};
            \addplot[fill=bar-green] coordinates    {(45,   193670)};
            \addplot[fill=bar-green] coordinates    {(50,   192173)};
            \addplot[fill=bar-green] coordinates    {(55,   188691)};
            \addplot[fill=bar-green] coordinates    {(60,   191683)};
            \addplot[fill=bar-green] coordinates    {(65,   192588)};
            \addplot[fill=bar-green] coordinates    {(70,   185398)};
            \addplot[fill=bar-green] coordinates    {(75,   186273)};
            \addplot[fill=bar-green] coordinates    {(80,   184274)};
            \addplot[fill=bar-green] coordinates    {(85,   188274)};
            \addplot[fill=bar-green] coordinates    {(90,   190283)};
            \addplot[fill=bar-green] coordinates    {(95,   189274)};
            \addplot[fill=bar-green] coordinates    {(100,   187245)};
            \addplot[fill=bar-green] coordinates    {(105,   186132)};
            \addplot[fill=bar-green] coordinates    {(110,   186824)};
            \addplot[fill=bar-green] coordinates    {(115,   184256)};
            \addplot[fill=bar-green] coordinates    {(120,   187842)};
            \addplot[fill=bar-green] coordinates    {(125,   187264)};
            \addplot[fill=bar-green] coordinates    {(130,   182748)};
            \addplot[fill=bar-green] coordinates    {(135,   186183)};
        \end{axis}
    \end{tikzpicture}
}

\newcommand{\graphRCC}{
    \definecolor{bar-green}{RGB}{105,200,76}
    \definecolor{bar-yellow}{RGB}{250,183,51}
    \definecolor{bar-orange}{RGB}{255,142,21}
    \definecolor{bar-red}{RGB}{255,13,13}
    \begin{tikzpicture}[plot, scale=1]
        \begin{axis}[
            scale=0.75,
                title style={font=\LARGE},
                label style={font=\LARGE},
                ybar,
                bar shift=2.5pt,
                bar width=5pt,
                scaled y ticks = false,
                xlabel={Time(s)},
                ytick={0, 50000, 100000, 150000, 200000, 250000, 300000},
                yticklabels={0, 50K,100K,150K,200K,250K,300K},
                xtick={0, 20, 40, 60, 80, 100, 120, 140},
                title={\RCC{} with $1$ failure},
                ymax=280000,
                ymin=0,
            ]
            \addplot[fill=bar-green] coordinates    {(0,   134953)};
            \addplot[fill=bar-green] coordinates    {(5,   136882)};
            \addplot[fill=bar-yellow] coordinates    {(10,   29940)};
            \addplot[fill=bar-yellow] coordinates    {(15,   20320)};
            \addplot[fill=bar-yellow] coordinates    {(20,   40640)};
            \addplot[fill=bar-yellow] coordinates    {(25,   0)};
            \addplot[fill=bar-yellow] coordinates    {(30,   81280)};
            \addplot[fill=bar-green] coordinates    {(35,   140299)};
            \addplot[fill=bar-yellow] coordinates    {(40,   22255)};
            \addplot[fill=bar-green] coordinates    {(45,   245719)};
            \addplot[fill=bar-yellow] coordinates    {(50,   78837)};
            \addplot[fill=bar-yellow] coordinates    {(55,   86139)};
            \addplot[fill=bar-green] coordinates    {(60,   287848)};
            \addplot[fill=bar-green] coordinates    {(65,   126600)};
            \addplot[fill=bar-green] coordinates    {(70,   122980)};
            \addplot[fill=bar-yellow] coordinates    {(75,   25800)};
            \addplot[fill=bar-green] coordinates    {(80,   171083)};
            \addplot[fill=bar-green] coordinates    {(85,   171083)};
            \addplot[fill=bar-green] coordinates    {(90,   194453)};
            \addplot[fill=bar-green] coordinates    {(95,   207269)};
            \addplot[fill=bar-green] coordinates    {(100,   131708)};
            \addplot[fill=bar-green] coordinates    {(105,   122398)};
            \addplot[fill=bar-green] coordinates    {(110,   130637)};
            \addplot[fill=bar-green] coordinates    {(115,   160234)};
            \addplot[fill=bar-green] coordinates    {(120,   132259)};
            \addplot[fill=bar-yellow] coordinates    {(125,   50137)};
            \addplot[fill=bar-yellow] coordinates    {(130,   67730)};
            \addplot[fill=bar-green] coordinates    {(135,   221440)};
        \end{axis}
    \end{tikzpicture}
}

\newcommand{\graphSpotLessf}{
    \definecolor{bar-green}{RGB}{105,200,76}
    \definecolor{bar-yellow}{RGB}{250,183,51}
    \definecolor{bar-orange}{RGB}{255,142,21}
    \definecolor{bar-red}{RGB}{255,13,13}
    \begin{tikzpicture}[plot, scale=1]
        \begin{axis}[
                scale=0.75,
                title style={font=\LARGE},
                label style={font=\LARGE},
                ybar,
                bar shift=2.5pt,
                bar width=5pt,
                scaled y ticks = false,
                xlabel={Time(s)},
                ylabel={Throughput(txn/s)},
                ytick={0, 50000, 100000, 150000, 200000, 250000, 300000},
                yticklabels={0, 50K,100K,150K,200K,250K,300K},
                xtick={0, 20, 40, 60, 80, 100, 120, 140},
                title={\SpotLess{} with $\f$ failures},
                ymax=280000,
                ymin=0,
            ]
            \addplot[fill=bar-green] coordinates    {(0,   184353)};
            \addplot[fill=bar-green] coordinates    {(5,   186352)};
            \addplot[fill=bar-green] coordinates    {(10,   117820)};
            \addplot[fill=bar-green] coordinates    {(15,   114390)};
            \addplot[fill=bar-green] coordinates    {(20,   109480)};
            \addplot[fill=bar-green] coordinates    {(25,   113900)};
            \addplot[fill=bar-green] coordinates    {(30,   109880)};
            \addplot[fill=bar-green] coordinates    {(35,   108800)};
            \addplot[fill=bar-green] coordinates    {(40,   118236)};
            \addplot[fill=bar-green] coordinates    {(45,   109552)};
            \addplot[fill=bar-green] coordinates    {(50,   103700)};
            \addplot[fill=bar-green] coordinates    {(55,   113900)};
            \addplot[fill=bar-green] coordinates    {(60,   103700)};
            \addplot[fill=bar-green] coordinates    {(65,   113900)};
            \addplot[fill=bar-green] coordinates    {(70,   108800)};
            \addplot[fill=bar-green] coordinates    {(75,   104614)};
            \addplot[fill=bar-green] coordinates    {(80,   102784)};
            \addplot[fill=bar-green] coordinates    {(85,   108800)};
            \addplot[fill=bar-green] coordinates    {(90,   103769)};
            \addplot[fill=bar-green] coordinates    {(95,   108719)};
            \addplot[fill=bar-green] coordinates    {(100,   109650)};
            \addplot[fill=bar-green] coordinates    {(105,   110137)};
            \addplot[fill=bar-green] coordinates    {(110,   105798)};
            \addplot[fill=bar-green] coordinates    {(115,   109590)};
            \addplot[fill=bar-green] coordinates    {(120,   112810)};
            \addplot[fill=bar-green] coordinates    {(125,   110213)};
            \addplot[fill=bar-green] coordinates    {(130,   110086)};
            \addplot[fill=bar-green] coordinates    {(135,   103972)};
        \end{axis}
    \end{tikzpicture}
    
}

\newcommand{\graphRCCf}{
    \definecolor{bar-green}{RGB}{105,200,76}
    \definecolor{bar-yellow}{RGB}{250,183,51}
    \definecolor{bar-orange}{RGB}{255,142,21}
    \definecolor{bar-red}{RGB}{255,13,13}
    \begin{tikzpicture}[plot, scale=1]
        \begin{axis}[
                scale=0.75,
                title style={font=\LARGE},
                label style={font=\LARGE},
                ybar,
                bar shift=2.5pt,
                bar width=5pt,
                scaled y ticks = false,
                xlabel={Time(s)},
                ytick={0, 50000, 100000, 150000, 200000, 250000, 300000},
                yticklabels={0, 50K,100K,150K,200K,250K,300K},
                xtick={0, 20, 40, 60, 80, 100, 120, 140},
                title={\RCC{} with $\f$ failures},
                ymax=280000,
                ymin=0,
            ]
            \addplot[fill=bar-green] coordinates    {(0,   135382)};
            \addplot[fill=bar-green] coordinates    {(5,   132734)};
            \addplot[fill=bar-yellow] coordinates    {(10,   13040)};
            \addplot[fill=bar-yellow] coordinates    {(15,   13760)};
            \addplot[fill=bar-yellow] coordinates    {(20,   27520)};
            \addplot[fill=bar-yellow] coordinates    {(25,   0)};
            \addplot[fill=bar-yellow] coordinates    {(30,   55040)};
            \addplot[fill=bar-yellow] coordinates    {(35,   72883)};
            \addplot[fill=bar-yellow] coordinates    {(40,   37160)};
            \addplot[fill=bar-green] coordinates    {(45,   162742)};
            \addplot[fill=bar-yellow] coordinates    {(50,   57416)};
            \addplot[fill=bar-green] coordinates    {(55,   128070)};
            \addplot[fill=bar-green] coordinates    {(60,   250933)};
            \addplot[fill=bar-yellow] coordinates    {(65,   61288)};
            \addplot[fill=bar-green] coordinates    {(70,   165739)};
            \addplot[fill=bar-green] coordinates    {(75,   239370)};
            \addplot[fill=bar-green] coordinates    {(80,   104548)};
            \addplot[fill=bar-green] coordinates    {(85,   104960)};
            \addplot[fill=bar-green] coordinates    {(90,   124785)};
            \addplot[fill=bar-green] coordinates    {(95,   141160)};
            \addplot[fill=bar-yellow] coordinates    {(100,   0)};
            \addplot[fill=bar-green] coordinates    {(105,   160904)};
            \addplot[fill=bar-green] coordinates    {(110,   249108)};
            \addplot[fill=bar-green] coordinates    {(115,   131516)};
            \addplot[fill=bar-green] coordinates    {(120,   129000)};
            \addplot[fill=bar-green] coordinates    {(125,   164380)};
            \addplot[fill=bar-green] coordinates    {(130,   135979)};
            \addplot[fill=bar-green] coordinates    {(135,   139489)};
        \end{axis}
    \end{tikzpicture}
}

\pgfplotstableread{
Tspotless    Lspotless   Trcc    Lrcc 
130072  1.180402  56049	 2.671848	
162066	1.930891  70038	 4.618901	
179107	3.482709  90708	 7.267852	
181225	6.779057  102322 13.066572	
184151  13.938311 109119 23.123579
}\dataTputLatFone
					
\pgfplotstableread{
Tspotless    Lspotless      Trcc    Lrcc
96963	2.075085 47019  2.15229
112428	2.996812 69102  3.320428
116232  3.725845 81208	5.686552	
119034	7.252736 88898	10.275836	
119389  13.353943 91037	19.70213
}\dataTputLatFF

\pgfplotstableread{
inflight spotless spotless1 spotlessf1 rcc rcc1 rccf
12  129384 130072	96963      98111		56049   47019
25 162955 	162066	112428	   120801      70038   69102  
50 175307 	179107	116232	   138198      90708   81208 
100 187966 	181225	119034	   146379      102322  88898   
200 187734 184151   119389     142659      109119  91037
}\dataFlightTput

\pgfplotstableread{
inflight spotless spotless1 spotlessf1 rcc rcc1 rccf
12  1.216809  1.180402	2.075085  1.641854  2.671848  2.15229		
25 1.952068  1.930891	2.996812  2.778643  4.618901  3.320428
50 3.593024  3.482709	3.725845  4.629623  7.267852  5.686552 
100 6.456446  6.779057	7.252736  8.545887  13.066572 10.27583
200 14.010092 13.938311 13.353943 17.408783 23.123579 19.70213
}\dataFlightLat

\newcommand{\resultgraphtputlatF}[5]{\begin{tikzpicture}[plot]
    \begin{axis}[xlabel={#3},ylabel={#4},title={#2},xtick={#5},xmin={0},xmax={200000}, ymin={0}, ymax={30}, legend to name={mainlegend5},legend columns=-1,
		grid=major]
        \addplot table[x={Tspotless},y={Lspotless}] {#1};
        \addplot table[x={Trcc},y={Lrcc}] {#1};
        \legend{\SpotLess{}, \RCC{}}
    \end{axis}
\end{tikzpicture}}

\newcommand{\flighttput}[5]{\begin{tikzpicture}[plot]
    \begin{axis}[xlabel={#3},ylabel={#4},title={#2},xtick={#5},xmin={0},xmax={}, ymin={0}, ymax={}, legend to name={mainlegend9},legend columns=-1,
		grid=major]
        \addplot table[x={inflight},y={spotless}] {#1};
        \addplot table[x={inflight},y={spotless1}] {#1};
        \addplot table[x={inflight},y={spotlessf1}] {#1};
        \addplot table[x={inflight},y={rcc}] {#1};
        \addplot table[x={inflight},y={rcc1}] {#1};
        \addplot table[x={inflight},y={rccf}] {#1};
        \legend{\SpotLess{}, $\SpotLess{}_{1}$, $\SpotLess{}_{f}$, \RCC{}, $\RCC{}_1$, $\RCC{}_f$}
    \end{axis}
\end{tikzpicture}}

\pgfplotstableread{
resource spotless rcc pbft hotstuff narwhal
4   44261	29502	13414	2919	14959	
8   97865   63437	25776	4152	34360	
16  179362  132475  27984	4426	65243	
32  230186  177133  27873   4064    131913
}\dataTputCore
				
\pgfplotstableread{
resource spotless rcc pbft hotstuff narwhal
500  66509  69295   6277    2742    55393
1000 134851 118639  12203   3643    57291
2000 164087 119850  25801   4013    57424
3000 183404 121420  26113   4224    57650
4000 184354	122034  27984   4426	58243
}\dataTputBandwidth
	
\pgfplotstableread{
resource spotless rcc pbft hotstuff narwhal
1 179362   132475   27984   4426    65243	
2 162133   120454   27685	1041	57688	
3 140647   115275   25910	715	    51502	
4 123836   114797   25441   695	    41663		
}\dataTputRegion

\pgfplotstableread{
resource spotless rcc pbft hotstuff narwhal
1 257327   233300   42289 13049 120459   			
2 249101   228238   41450 2143  112166
3 242189   226587   40240 1856  106740	
4 230327   220422   39116 1465  97780
}\dataTputRegionFour

\pgfplotstableread{
Failures hs1 hs2 hs3 hs4 spotless1 spotless2 spotless3 spotless4
0 4833  4833    4833    4833    5874  5874	5874    5874
1 4523  4872    4812    4874    5721  5827  5865    5851
2 4232  4802    4799    4798    5475  5893  5842	5852
3 4021  4827    4773    4823    5301  5891  5832	5813
4 3838  4888    4760    4914    5199  5901  5825    5825
6 3634  4824    4733    4823    4881  5968  5800	5847
8 3452  4851    4705    4893    4829  5888  5750    5851
10 3228 4825    4682    4833    4762  5854  5703    5851    
}\dataTputSingleInsatnceF
							
\pgfplotstableread{
Failures hs1 hs2 hs3 hs4 spotless1 spotless2 spotless3 spotless4
0   4833  4833    4833    4833  5874  5874  5874    5874
0.2 3343  4851    4705    4893  4829  5888  5750    5842
0.4 2290  4863    4688    4847  4423  5892	5649    5892		
0.6 1408  4842    4612    4823  4042  5850  5566    5838
0.8 987   4871    4525    4843  3745  5928  5360    5815
1   766   4835    4498    4872  3283  5888	5396    5883
}\dataTputSingleInsatnceFR

\pgfplotstableread{
Failures A1 A2 A3 A4
0   4833  4833    4833    4833 
0.2 3343  4851    4705    4893 
0.4 2290  4863    4688    4847 
0.6 1408  4842    4612    4823 
0.8 987   4871    4525    4843 
1   766   4835    4498    4872 
}\dataTputSingleInsatnceHS

\pgfplotstableread{
Failures A1 A2 A3 A4
0   5874  5874  5874    5874
0.2 4829  5888  5750    5842
0.4 4423  5892	5649    5892		
0.6 4042  5850  5566    5838
0.8 3745  5928  5360    5815
1   3283  5888	5396    5883
}\dataTputSingleInsatncePVP

\newcommand{\resultgraphsifaiurer}[5]{\begin{tikzpicture}[plot]
    \begin{axis}[xlabel={#3},ylabel={#4},title={#2},xtick={#5}, ymax={7000}, ymin={0}, legend to name={si_legend},legend columns=-1, scaled y ticks=real:1000]
        \addplot table[x={Failures},y={A1}] {#1};
        \addplot table[x={Failures},y={A2}] {#1};
        \addplot table[x={Failures},y={A3}] {#1};
        \addplot table[x={Failures},y={A4}] {#1};
        \legend{A1, A2, A3, A4}
    \end{axis}
\end{tikzpicture}}

\newcommand{\axistput}{Throughput (\si{\txn\per\second})}
\newcommand{\axislat}{Latency (\si{\second})}
\newcommand{\axisnodes}{Number of replicas ($\n$)}
\newcommand{\axisregions}{Number of regions}
\newcommand{\axisbandwidth}{Network Bandwidth (Mbit/s)}
\newcommand{\axiscore}{Number of CPU Cores}
\newcommand{\axisinstances}{Number of concurrent instances}
\newcommand{\axisfailures}{Number of malicious replicas}
\newcommand{\axisbfailures}{Number of malicious replicas}
\newcommand{\axispercentage}{Ratio of malicious replicas (out-of-$\f$)}
\newcommand{\axisbpercentage}{Ratio of Byzantine failures (out-of-$\f$)}
\newcommand{\axisbatches}{Batch size}
\newcommand{\axistxnsize}{Transaction size ($\si{\byte}$)}
\newcommand{\axisinflight}{Client Batches per Primary receives}
\newcommand{\axisticksnodes}{4,16,32,64,96,128}
\newcommand{\axisticksbatches}{10,50,100,200,400}
\newcommand{\axisticksinstances}{1,8,16,32,64}
\newcommand{\axisticksinstancesb}{1,16,32,64,128}
\newcommand{\axisticksfailures}{0,1,2,3,4,6,8,10}
\newcommand{\axistickspercentage}{0,0.2,0.4,0.6,0.8,1}
\newcommand{\axistickstxnsize}{48,200,400,600,800,1600}
\newcommand{\axistickstput}{0,50000,100000,150000,200000,250000,300000}
\newcommand{\axisticksinflight}{12, 25, 50, 100, 200}
\newcommand{\axisticksregions}{1, 2, 3, 4}
\newcommand{\axisticksbandwidth}{500, 1000, 2000, 3000, 4000}
\newcommand{\axistickscores}{4,8,16,32}

\title[SpotLess: Concurrent Rotational Consensus Made Practical through Rapid View Synchronization]{\Changed{SpotLess: Concurrent Rotational Consensus Made Practical through Rapid View Synchronization}}

\author{Dakai Kang, Sajjad Rahnama, Jelle Hellings$^{\dagger}$, Mohammad Sadoghi}
\affiliation{
    \institution{Exploratory Systems Lab, Department of Computer Science, University of California, Davis}
	\institution{$\dagger$Department of Computing and Software, McMaster University}
}

\begin{document}

\begin{abstract}
    The emergence of blockchain technology has renewed the interest in \emph{consensus-based} data management systems that are resilient to failures. To maximize the throughput of these systems, we have recently seen several prototype consensus solutions that optimize for throughput at the expense of overall implementation complexity, high costs, and reliability. Due to this, it remains unclear how these prototypes will perform in real-world environments.
    
    In this paper, we present \SpotLess{}, a novel concurrent rotational consensus protocol made practical. Central to \SpotLess{} is the combination of (1) \Changed{a \emph{chained rotational consensus design} for replicating requests with a reduced message cost and low-cost failure recovery that eliminates the traditional complex, error-prone view-change protocol; (2) the novel \emph{Rapid View Synchronization} protocol that enables \SpotLess{} to work in more general network assumptions, without a need for a \emph{Global Synchronization Time} to synchronize view, and recover valid earlier views with the aid of non-faulty replicas without the need to rely on the primary; (3) a high-performance \emph{concurrent consensus architecture} in which independent instances of the chained consensus operate concurrently to process requests with high throughput, thereby avoiding the bottlenecks seen in other rotational protocols.}
    
    Due to the concurrent consensus architecture, \SpotLess{} greatly outperforms traditional primary-backup consensus protocols such as \PBFT{} (by up to 430\%), \Narwhal{} (by up to 137\%), and \HS{} (by up to 3803\%). Due to its reduced message cost, \SpotLess{} is even able to outperform \RCC{}, a state-of-the-art high-throughput concurrent consensus protocol, by up to 23\%. Furthermore, \SpotLess{} is able to maintain a stable and low latency and consistently high throughput even during failures.
    \end{abstract}

\maketitle

\pagestyle{\vldbpagestyle}


\section{Introduction}

The emergence of \Bitcoin{}~\cite{bitcoin} and blockchain technology has renewed the interest in \emph{consensus-based} 
resilient data management systems (\RDMS{}s)\Changed{~\cite{aca_ped,pbftj,flp,byzgeneralproblem,pease1980reaching,dworkmodel}} that can provide resilience to failures and can manage data 
between fully-independent parties (federated data management). Due to these qualities, there is widespread 
interest in \RDMS{}s with applications in finance, health care, IoT, agriculture, fraud-prevention, and other 
industries~\cite{blockhealthover,blockchain_iot,blockfood,bceco,ruan2022}.

Although \Bitcoin{} builds on many pre-existing techniques, the novel way in which 
\Bitcoin{} used these techniques was a major breakthrough for resilient systems, as \Bitcoin{} showed 
that resilient systems with thousands of participants can solve large-scale problems~\cite{aca_ped,blockdist}. 
Furthermore, \Bitcoin{} did so in a \emph{permissionless} way without requiring a known set of participants 
and allowing participants to join and leave the system at any time. The highly-flexible permissionless design 
of blockchains such as \Bitcoin{} and \Ethereum{}~\cite{ethereum} is not suitable for high-performance \RDMS{}s, 
however: their abysmal transaction throughput, high operational costs, and per-transaction costs make 
them unsuitable for typical data-based applications~\cite{bitcoin,ethereum,aca_ped,blockdist,badcoin,badbadcoin,eyal2018majority}. Instead, 
data-based applications often are deployed in an environment with a set of identifiable \emph{participants} 
(who may behave arbitrarily) due to which they can use \emph{permissioned} designs using primary-backup consensus 
protocols~\cite{sbft,hotstuff,zyzzyvaj,pbftj,cbook} such as the \emph{Practical Byzantine Fault Tolerance consensus protocol} (\PBFT{})~\cite{pbftj}. 

\begin{figure*}
\setlength\tabcolsep{3pt}
\centering
\small
\begin{tabular}{lccccccccc} 
\toprule
&\multicolumn{2}{c}{Environment}&Concurrent&Chained&Threshold&\multicolumn{4}{c}{Communication Complexity}\\
Protocol&Safety&Liveness&Consensus&Consensus&Signatures&Phases&Messages&(\emph{at primary})&(\emph{per decision})\\
\midrule
\SpotLess{} &\Good{}Asynchronous&\Good{}Partial Synchrony&\Good{}yes&\Good{}yes&\Good{}no&$6$&$c(3\n^2)$&$c(3\n)$&$\n^2$\\
\midrule
\PBFT{}~\cite{pbftj}&\Good{}Asynchronous&\Good{}Partial Synchrony&no&no&\Good{}no&$3$&$2\n^2$&$3\n$&$2\n^2$\\
\RCC{}~\cite{rcc}&\Good{}Asynchronous&\Good{}Partial Synchrony&\Good{}yes&no&\Good{}no&$3$&$c( 2\n^2)$&$c(3\n)$&$2\n^2$\\
\HS{}~\cite{hs}&\Good{}Asynchronous&Partial Synchrony&no&\Good{}yes&yes&$8$&$8\n$&$4\n$&$2\n$\\
\bottomrule
\end{tabular}
\caption{Comparison of \SpotLess{} with three state-of-the-art consensus protocols. Here, $\n$ is the number of replicas, $c$, $1 \leq c \leq \n$, is the number of concurrent instances, and the \emph{per decision} cost is the amortized cost of a single consensus decision.}
\label{fig:summary}
\end{figure*}

\RDMS{}s with fine-tuned primary-backup consensus implementations can process hundreds-of-thousands 
client requests per second~\cite{book}. Such high-throughput implementations come with severe limitations, 
however. First, in primary-backup consensus, a single replica (the primary) coordinates the replication of 
requests. Due to this central role of the primary, performance is usually bottlenecked by the network bandwidth 
or computational resources available to that primary~\cite{rcc,mirbft,iss}. Furthermore, the central role of the 
primary is detrimental to scalability, due to which high throughput can only be achieved on small-scale deployments. 
Finally, the techniques necessary in primary-backup consensus to reach high throughput 
(e.g., \emph{out-of-order processing}~\cite{pbftj,book}) \Changed{require complex implementations that keep track of many partially-processed rounds of consensus. When recovering from failures, this normal-case complexity necessitates complex and costly (in terms of message size and duration) \emph{view-change protocols}  to figure out which of these partially-processed consensus rounds can contribute to a consistent recovered state.} Recently, we have seen two 
significant developments to address these limitations in isolation.

First, the introduction of \emph{concurrent consensus protocols} such as \RCC{}~\cite{rcc},
\MirBFT{}~\cite{mirbft} and \ISS{}~\cite{iss} have significantly improved the scalability and performance of high-throughput consensus. 
These concurrent consensus protocols do so by taking a primary-backup consensus protocol such as \PBFT{} as 
their basis and then run multiple instances (each with a distinct primary, e.g., each non-faulty replica is a 
primary of its own instance) at the same time, this to remove any single-replica bottlenecks. On the one hand, 
the concurrent consensus is able to eliminate bottlenecks, improve scalability, and improve performance. On the 
other hand, existing concurrent consensus protocols do so by further increasing 
both the implementation complexity and the cost of \emph{recovery}. \Changed{For example, \RCC{} shuts down faulty primaries for an exponentially increasing number of rounds after receiving sufficient complaints.}

Second, the introduction of the \emph{chained consensus protocol} \HS{}~\cite{hs} has provided a simplified 
and easier-to-implement consensus protocol with low communication costs. To achieve this, \HS{} chains 
consecutive consensus decisions, which allows \HS{} to overlap communication costs for consecutive consensus 
decisions and minimize the cost of \emph{recovery}. \HS{} uses low-cost recovery to change 
primaries after each consensus decision, thereby reducing the impact of any malicious primaries. Finally, 
\HS{} uses \emph{threshold signatures}~\cite{tsecc} to make all communication phases \emph{linear} in cost. 
The commendable simplicity and low cost of \HS{} do come at the expense of performance and resilience, 
however. First, \Changed{the rotational design of \HS{}, which disables \emph{out-of-order processing}, inherently bounds performance by \emph{message delays} and makes \HS{} incapable of fully utilizing computing and network resources, which causes the low throughput of \HS{}.} The negative impact of message delays is further compounded by the reliance on threshold signatures, which incur additional rounds of communication and have high computational costs.  \Changed{Furthermore, the low-cost design of \emph{recovery} in \HS{} reduces the resilience compared to \PBFT{}, as \HS{} relies on a black-box Pacemaker for \emph{view synchronization}, which is essential to the liveness of rotational protocols.}~\cite{hslive, prestigebft}.

\Changed{In this paper, we present \SpotLess{}, the first practical consensus protocol that combines simplicity with high performance. \SpotLess{} does so by combining a novel \emph{chained rotational consensus design} that is optimized toward simplicity, resilience, low message complexity, and latency with a high-performance \emph{concurrent consensus architecture}. Central to the chained rotational consensus design of \SpotLess{} is \emph{Rapid View Synchronization}  (\RVS{}), which provides continuous low-cost primary rotation to deal with malicious behavior. \RVS{} enables \SpotLess{} to work in more general network assumptions, without a need for a \emph{Global Synchronization Time} to synchronize view, and recover valid earlier views with the aid of non-faulty replicas without the need to rely on the primary.

The rotational design of \SpotLess{} eliminates the need for the traditional complex and error-prone view-change protocols found in \PBFT{} and its variants: due to the rotational design of \SpotLess{}, only information on a single round is used during recovery. In addition, \RVS{} provides strong \emph{view synchronization}, resolving the liveness issues of previous works. Furthermore, \RVS{} does not require costly threshold signatures and provides robust failure recovery steps even when communication is unreliable. Finally, by combining the \emph{chained rotational consensus design} with a \emph{concurrent consensus architecture}, we remove the bottleneck of message delays typically seen in rotational designs \emph{without} having to resort to highly-complex implementation techniques such as out-of-order processing.}

To evaluate the performance of \SpotLess{} in practice, we 
have implemented \SpotLess{} in \Changed{\RDB{} (Incubating)},
our high-performance resilient blockchain database that serves as a testbed for 
future \RDMS{} technology. Our evaluation shows that \SpotLess{} greatly outperforms existing consensus 
protocols such as \PBFT{}~\cite{pbftj} by up to 430\%, \Narwhal{}~\cite{narwhal} by up to 137\%, 
and \HS{}~\cite{hotstuff} by up to 3803\%. Furthermore, due to the low 
message complexity of \SpotLess{}, it is even able to outperform \RCC{}~\cite{rcc} by up to 23\% in normal  
conditions while serving client requests with lower latency in all cases. Finally, due to the 
robustness of \RVS{}, \SpotLess{} is able to maintain a stable latency and consistently high throughput even 
during failures.

Our contributions are as follows:

\begin{enumerate}
\item In Section~\ref{sec:design}, we present the single-instance \emph{chained consensus design} of \SpotLess{} 
that provides the consensus replication using \emph{rapid view synchronization}. 
\item In Section~\ref{sec:concurrent}, we provide the \emph{concurrent consensus architecture} employed by 
\SpotLess{} to run multiple instances of the chained consensus \emph{in parallel}, due to which \SpotLess{} has 
highly-scalable throughput akin to \RCC{}.
\item In Section~\ref{sec:eval}, we empirically evaluate \SpotLess{} in \Changed{\RDB{}} and compare its performance with 
state-of-the-art consensus protocols such as \HS{}~\cite{hotstuff}, \PBFT{}~\cite{pbftj}, \RCC{}~\cite{rcc}, 
and \Narwhal{}~\cite{narwhal}. In our evaluation, we show the \emph{excellent properties} of \SpotLess{}, which 
is even able to achieve higher throughput than the concurrent consensus protocol \RCC{}, while providing 
a low and stable latency in all cases.
\end{enumerate}

In addition, we introduce the terminology and notation used throughout this paper in 
Section~\ref{sec:prelim}, discuss related work in Section~\ref{sec:related}, and conclude on our 
findings in Section~\ref{sec:conclusion}. Finally, we have summarized the properties of \SpotLess{} and 
how they compare with other common and state-of-the-art consensus protocols in Figure~\ref{fig:summary}.

\section{Preliminaries}\label{sec:prelim}

\paragraph{System}
We model our \emph{system} as a fixed set of replicas $\Replicas$.  We write $\n = \abs{\Replicas}$ to denote 
the number of replicas and we write $\f$ to denote the number of \emph{faulty replicas}. Each replica 
$\Replica \in \Replicas$ has a unique identifier $\ID{\Replica}$ with $0 \leq \ID{\Replica} < \n$.  We assume $\n > 3\f$ 
(a minimal requirement to provide consensus in an asynchronous environment~\cite{book}), that non-faulty replicas behave in accordance with the protocols they are executing, and that faulty replicas can behave arbitrarily, possibly coordinated and malicious ways. We do not make 
any assumptions about clients: all clients can be malicious without affecting \SpotLess{}.

\paragraph{Consensus}
\SpotLess{} is a \emph{consensus protocol} that decides the sequence of \emph{client requests} executed 
by all non-faulty replicas in the system $\Replicas$. To do so, \SpotLess{} provides three 
\emph{consensus guarantees}~\cite{book,distalgo}:
\begin{enumerate}
    \item \emph{Termination}. If non-faulty replica $\Replica \in \Replicas$ decides upon an $\rn$-th 
    client request, then all non-faulty replicas $\Replica[q] \in \Replicas$ will decide upon an 
    $\rn$-th client request;
    \item \emph{Non-Divergence}. If non-faulty replicas $\Replica_1, \Replica_2 \in \Replicas$ make 
    $\rn$-th decisions $\T_1$ and $\T_2$, respectively, then $\T_1 = \T_2$ (they decide upon the same 
    $\rn$-th client request).
    \item \emph{Service.} Whenever a non-faulty client $\Client$ requests execution of $\T$, then all 
    non-faulty replicas will eventually decide on a client request of $\Client$.
\end{enumerate}
We note that we use \SpotLess{} in a setting of a replicated service that executes client requests. Hence, 
instead of the abstract \emph{non-triviality guarantee} typically associated with consensus~\cite{book}, \SpotLess{} 
guarantees \emph{service}. Adapting \SpotLess{} to settings where other versions of \emph{non-triviality} 
are required is straightforward.

\paragraph{Communication}
We assume \emph{asynchronous communication}: messages can get lost or arbitrarily delayed. As consensus cannot 
be solved in asynchronous environments~\cite{flp}, we adopt the partial synchrony model of \PBFT{}~\cite{pbftj}: we 
\emph{always} guarantee non-divergence (referred to as \emph{safety}), while only guaranteeing 
termination and service during periods of reliable communication with a bounded message delay (referred 
to as \emph{liveness}). We assume that periods of unreliable communication are always followed by sufficiently-long 
periods of synchronous communication (during which \SpotLess{} can complete a limited number of steps to 
restore liveness).

\paragraph{Authentication}
We assume \emph{authenticated communication}: faulty replicas are able to impersonate each other, 
but replicas cannot impersonate non-faulty replicas. Authenticated communication is a minimal requirement 
to deal with malicious behavior~\cite{book}. To enforce authenticated communication, we use two mechanisms: 
\emph{message authentication codes} (\MAC{}s) and \emph{digital signatures} (\DS{}s)~\cite{cryptobook}. As 
\MAC{}s do not guarantee \emph{tamper-free message forwarding}, we only use \MAC{}s (which are cheaper than 
\DS{}s) to authenticate those messages that are not forwarded. For all other messages, we use \DS{}s. 
We write $\Cert{v}{p}$ to denote a value $v$ signed by participant $p$ (a client or a replica). 
Finally, we write $\digest{v}$ to denote the message digest of a value $v$ constructed using the same 
\emph{secure cryptographic hash function} as the one used when signing $v$~\cite{cryptobook}.

\section{SpotLess Design Principles}\label{sec:design}
\SpotLess{} combines a chained consensus design with
a high-performance concurrent architecture. To maximize resilience in practical network environments in which
communication can become \emph{unreliable} \Changed{and messages can get lost}, the chained consensus instances of \SpotLess{} use \emph{Rapid View Synchronization} (\RVS{}) to assure that each instance can always recover and resume consensus.

Our presentation of individual \SpotLess{} instance is broken up into \Changed{five} parts. \Changed{Each \emph{chained consensus instance} of \SpotLess{} operates in views $\v \gets 0, 1, 2, 3, \dots$. First, in Section~\ref{ss:single}, we show the two steps in every view. Second, in Section~\ref{ss:commit}, we present the normal-case replication steps and the three-phase commit algorithm used by each chained consensus instance. Third, in Section~\ref{ss:rules}, we formalize the guarantees provided by the normal-case replication steps and prove the safety of \SpotLess{}.} Then, in Section~\ref{ss:rvs}, we present the design of \emph{\RVS{}}. \RVS{} bootstraps the guarantees provided by the normal-case replication toward providing per-instance consensus. Next, in Section~\ref{ss:timer}, we describe how \SpotLess{} assures per-instance consensus in an asynchronous environment and formally prove the liveness of \SpotLess{}.

\subsection{\Changed{Steps in Every View: Propose and Synch Primitives}}\label{ss:single}

\Changed{View $\v$ is coordinated by the replica $\Primary{} \in \Replicas$ with
$\ID{\Primary{}} = \v \bmod \n$.} We say that $\Primary{}$ is the \emph{primary} of view $\v$ and all
other replicas act as \emph{backups}. In view $\v$, primary $\Primary{}$ will be
able to \emph{propose} the next client request $\T$ upon which the system aims to achieve consensus. \Changed{
To do so, the system proceeds in \emph{two steps i.e. Propose and Synch}. First, the primary inspects the existing chain and decides from which proposal it extends a new proposal, then the primary picks a valid client request $\T$, wraps and broadcasts a new proposal and broadcasts. Second, the backup replicas decide whether to vote for the new proposal and broadcast their decisions, where the new proposal is \emph{conditionally prepared} by a replica if it receives $\n-\f$ concurring votes. Now, we explain the two steps in detail below.}

First, primary $\Primary{}$ inspects the results of the preceding views to determine the highest extendable
proposal \Changed{$\PP'$ past} view $\v-1$\Changed{, such that $\Primary{}$ believes that at least $\n-\f$ replicas will vote for a new proposal extending $\PP'$. We will explore how $\PP'$ is chosen in Section~\ref{ss:rules}.} \Changed{Then, primary $\Primary{}$ picks a client request $\T$ from some client $\Client$ that it has not yet proposed
and proposes $\T$ by broadcasting a \MName{Propose} message of the form
$\Changed{\PP} \GETS \Message{Propose}{\v,\T,\PPproof{\Changed{\PP'}}}$ to all backups, in which $\PPproof{\Changed{\PP'}}$ is a \emph{certificate} for the preceding proposal $\Changed{\PP'}$ that $\Primary{}$ chooses. The certificate is either a list of $\n-\f$ digital signatures, and we will explain how certificates are used in Section~\ref{ss:rules}.} To assure that $\T$ cannot be forged by the primary, we assume that all client requests are digitally signed by the client $\Client$. To assure that the authenticity of $\Changed{\PP}$ can be established and that $\Changed{\PP}$ can be forwarded, the primary $\Primary{}$ will digitally sign the message $\Changed{\PP}$.

\Changed{Second,} the backups establish whether the primary $\Primary{}$ correctly proposed a \emph{unique} proposal to
them. Specifically, the backups will exchange \MName{Sync} messages between them via which they can determine
whether $\Changed{\PP}$ is the only proposal that can collect enough endorsements to generate a certificate in the current view (necessary to provide non-divergence) and to ensure that enough non-faulty replicas received the same well-formed 
proposal $\Changed{\PP}$ to assure that $\Changed{\PP}$ can be recovered in any future view (independent of any malicious behavior). To do so,
each backup $\Replica \in \Replicas$ performs the following steps upon receiving message
$\Changed{\PP} \GETS \Message{Propose}{\v,\T,\PPproof{\Changed{\PP'}}}$ with digital signature $\Cert{\Changed{\PP'}}{\Primary{}}$:
\begin{enumerate}
\renewcommand{\theenumi}{S\arabic{enumi}}
\item $\Replica$ checks whether $\Cert{\Changed{\PP'}}{\Primary{}}$ is a \emph{valid digital signature};
\item $\Replica$ checks whether $\T$ is a \emph{valid client request};
\item $\Replica$ checks whether view $\v$ is the \emph{current} view; and
\item\label{s:claim} \Changed{$\Replica$ checks whether $\PPproof{\Changed{\PP'}}$ is valid if
$\Replica$ has not \emph{conditionally prepared} $\Changed{\PP'}$.}
\end{enumerate}

Only if the proposal $\Changed{\PP}$ passes all these checks, the replica $\Replica$ will consider $\Changed{\PP}$ to be \emph{well-formed}. In this case, backup $\Replica$ \emph{records} $\Changed{\PP}$. If $\Changed{\PP}$ is the first \emph{acceptable} proposal $\Replica$ receives in view $\v$ (we detail the conditions of \emph{acceptable} proposals in Section~\ref{ss:rules}), $\Replica$ will broadcast the message $\RVar{ms}{\Replica} \GETS \Message{Sync}{\v, \PPf{\Changed{\PP}}, \AcceptedSet}$, in which $\PPf{\Changed{\PP}} \GETS (\v, \digest{\Changed{\PP}},\Cert{\Changed{\PP}}{\Primary{}})$ is a \emph{claim} that $\Changed{\PP}$ is the  well-formed proposal that backup $\Replica$ received in view $\v$, and \Changed{$\AcceptedSet$ is a set of pairs in the form of (view, digest) for the proposals that $\Replica$ has \emph{conditionally prepared}}. We will explore the details of $\AcceptedSet$ in Section~\ref{ss:rules}.

Otherwise, if backup $\Replica$ determines a failure in view $\v$ ($\Replica$ did not receive any valid proposals in view $\v$ while it should have received one), then $\Replica$ will end up broadcasting the message $\RVar{ms}{\Replica} \GETS \Message{Sync}{\v, \PPf{\varnothing}, \AcceptedSet}$ to all backup replicas, \Changed{claiming to have not received any valid proposals in view $v$.}

To assure that the authenticity of $\RVar{ms}{\Replica}$ can be established \emph{without verifying digital
signatures} and that $\RVar{ms}{\Replica}$ can be forwarded, the replica $\Replica$ will include both a message
authentication code and the digital signature $\Cert{\RVar{ms}{\Replica}}{\Replica}$. To reduce computational
costs, the message authentication codes of \MName{Sync} messages are \emph{always} verified, whereas digital
signatures are \emph{only verified in cases where recovery is necessary}, we refer to Section~\ref{ss:rvs} for
the exact verification rules.

\begin{remark}\label{rem:redun}
To simplify presentation, all replicas broadcast messages to all replicas (including themselves). Doing so, we simplify the claims and proofs made in this paper. Without affecting the correctness of \SpotLess{}, one can
eliminate sending messages to oneself. In addition, the primary broadcasts the proposal $\Changed{\PP}$ together with a matching \MName{Sync} message.
\end{remark}

\subsection{\Changed{Chained Three-Phase Commit Algorithm in Normal Case}}\label{ss:commit}

\begin{figure}[t]
    \centering
    \begin{tikzpicture}[yscale=0.5]
        \draw[thick,draw=black!75] (1.75,   0) edge ++(6.5, 0)
                                   (1.75,   1) edge ++(6.5, 0)
                                   (1.75,   2) edge ++(6.5, 0)
                                   (1.75,   3) edge ++(6.5, 0);
        \path[thick,blue] (2, 3) edge (4, 3)
                          (4, 2) edge (6, 2)
                          (6, 1) edge (8, 1);

        \node[right,align=left,smalltext] (x) at (4.3,3.8)
            {Primary $\Replica_2$ determines proposal to extend\\
             based on incoming \MName{Sync} messages.};

        \draw[thin,draw=black!75] (2, 0) edge ++(0, 3)
                                  (4,   0) edge ++(0, 3)
                                  (6, 0) edge ++(0, 3)
                                  (8,   0) edge ++(0, 3);

        \draw[thin,draw=black!50] (3, 0) edge ++(0, 3)
                                  (5,   0) edge ++(0, 3)
                                  (7, 0) edge ++(0, 3);

        \node[left] at (1.8, 0) {$\Replica_4$};
        \node[left] at (1.8, 1) {$\Replica_3$};
        \node[left] at (1.8, 2) {$\Replica_2$};
        \node[left] at (1.8, 3) {$\Replica_1$};

        \path[->] (2, 3) edge (3, 0)
                         edge (3, 1)
                         edge (3, 2);

        \path[->] (3, 0) edge (4, 1) edge (4, 2) edge (4, 3)
                  (3, 1) edge (4, 0) edge (4, 2) edge (4, 3)
                  (3, 2) edge (4, 0) edge (4, 1) edge (4, 3);

        \path[->] (4, 2) edge (5, 0)
                         edge (5, 1)
                         edge (5, 3);

        \path[->] (5, 0) edge (6, 1) edge (6, 2) edge (6, 3)
                  (5, 1) edge (6, 0) edge (6, 2) edge (6, 3)
                  (5, 3) edge (6, 0) edge (6, 1) edge (6, 3);

        \path[->] (6, 1) edge (7, 0)
                         edge (7, 2)
                         edge (7, 3);

        \path[->] (7, 0) edge (8, 1) edge (8, 2) edge (8, 3)
                  (7, 2) edge (8, 0) edge (8, 1) edge (8, 3)
                  (7, 3) edge (8, 0) edge (8, 1) edge (8, 2);

        \node[dot,black] at (2,3) {};
        \node[dot,black] at (4,2) {} edge[thin,black] (4.3,3.8);
        \node[dot,black] at (6,1) {};

        \node[below,smalltext] at (2.5, 0) {\MName{Propose}};
        \node[below,smalltext] at (4.5, 0) {\MName{Propose}};
        \node[below,smalltext] at (6.5, 0) {\MName{Propose}};
        \node[below,smalltext] at (3.5, 0) {\MName{Sync}};
        \node[below,smalltext] at (5.5, 0) {\MName{Sync}};
        \node[below,smalltext] at (7.5, 0) {\MName{Sync}};

        \draw[decoration={brace,amplitude=0.5em,mirror},decorate,thick]
                (2.05, -0.6) -- node[below=3pt,align=center,smalltext] {view $\v-1$,\\proposal $\PP_1$,\\ \footnotesize{{\emph conditional prepare} $\PP_1$} \\ \footnotesize{$\PP' \gets \PP_1$}} (3.95, -0.6);
        \draw[decoration={brace,amplitude=0.5em,mirror},decorate,thick]
                (4.05, -0.6) -- node[below=3pt,align=center,smalltext] {view $\v$,\\proposal $\PP_2$,\\ \footnotesize{{\emph conditional commit} $\PP_1$} \\ \footnotesize{$\PP' \gets \PP_2$} \\ \footnotesize{$\PP_{lock} \gets \PP_1$}} (5.95, -0.6);
        \draw[decoration={brace,amplitude=0.5em,mirror},decorate,thick]
                (6.05, -0.6) -- node[below=3pt,align=center,smalltext] {view $\v+1$,\\proposal $\PP_3$,\\ \small{{\emph commit} $\PP_1$} \\ \footnotesize{$\PP' \gets \PP_3$}} (7.95, -0.6);
    \end{tikzpicture}
    \caption{\Changed{A schematic representation of the normal-case replication protocol in a chained consensus instance of \SpotLess{} in three consecutive views $\v - 1$ (primary $\Replica_1$), $\v$ (primary $\Replica_2$), and $\v + 1$ (primary $\Replica_3$). \Changed{$\PP'$ refers to the \emph{highest extendable proposal}}. We have \emph{not} visualized those messages one can eliminate (see Remark~\ref{rem:redun}).}}
\label{fig:normal_case}
\end{figure}
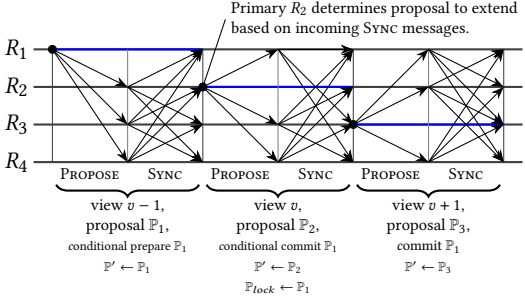

\Changed{\SpotLess{} adopts a three-phase commit algorithm. Each phase takes one view, and \SpotLess{} rotates the primary view by view, to eliminate complex failure detection and recovery.} A successful first phase of \SpotLess{} establishes a \emph{conditional prepare} that ensures non-divergence of proposals within the view; a successful second phase of \SpotLess{} establishes a \emph{conditional commit}; and a successful third phase achieves a \emph{commit} that ensures the preservation of proposals across views. We explain the conditions to establish these three proposal states (i.e., \emph{conditional prepare}, \emph{conditional commit}, and \emph{commit}) later in Definition~\ref{def:accept_lock}. In Figure~\ref{fig:normal_case}, we have visualized three \emph{consecutive} operations of a chained \Changed{rotational} consensus instance in views $\v-1$, $\v$, and $\v+1$ with respect to a proposal \Changed{$\Changed{\PP}$} and in Figure~\ref{alg:normal_case}, we present the pseudo-code of the normal-case operations of the chained consensus.

The normal-case replication protocol of \SpotLess{} only establishes a minimal guarantee on the overall state of the system that can be proven with a straightforward quorum-based argument~\cite{book}:
\begin{theorem}\label{thm:unique}

Consider view $\v$ of the normal-case replication of a chained consensus instance of \SpotLess{} and consider 
two replicas $\Replica_1, \Replica_2 \in \Replicas$. If $\n > 3\f$ and replica $\Replica_i$, $i \in \{1,2\}$, receives a set of authenticated messages $\{ \Message{Sync}{\v,\PPf{\PP{}_{i,\Replica[q]}}} \mid \Replica[q] \in \Replicas[q]_i \}$ from a set $\Replicas[q]_i \subseteq \Replicas$ of $\abs{\Replicas[q]_i} = \n - \f$ replicas, 
and all $\PPf{\PP{}_{i,\Replica[q]}}$, $i\in \{1,2\}$, represent proposal $\PP_i$, then $\PP_1 = \PP_2$.
\end{theorem}

\begin{proof}
    We prove the theorem by contradiction. Assume $\PP_1 \neq \PP_2$. Replica $\Replica_i$ received claims representing $\PP_i$ from the $\n-\f$ in  $\Replicas[q]_i$. As there are at most $\f$ faulty replicas, at-least $\n-2\f$ replicas in $\Replicas[q]_i$ are non-faulty. Let $S_i \subseteq \Replicas[q]_i$ be all non-faulty replicas in $\Replicas[q]_i$.
    
    As the non-faulty replicas in $S_1$ made claims representing $\PP_1$, the non-faulty replicas in $S_2$ made claims representing $\PP_2$, $\PP_1 \neq \PP_2$, and non-faulty replicas only make a single claim in view $v$ (see Figure~\ref{alg:normal_case}), we must have $S_1 \intersect S_2 = \emptyset$. Hence, we have $\abs{S_1 \union S_2} \geq 2(\n-2\f)$. As all replicas in  $ S_1 \union S_2$ are non-faulty, we also have $\abs{ S_1 \union S_2} \leq \n-f$. Hence, $2(\n - 2\f) \leq \n-\f$ must hold, which implies $\n \leq 3\f$, a contradiction. Hence, by contradiction, we conclude $\PP_1 = \PP_2$.
\end{proof}

\subsection{Rules Guaranteeing Safety}\label{ss:rules}

Beside the minimal guarantee in Theorem~\ref{thm:unique}, \SpotLess{} ensures the safety of the system via a set of rules
that non-faulty replicas must follow. In the meantime, the rules play an important role in helping restore liveness while
ensuring safety. Before exploring the rules, we first introduce the necessary terminology.

\begin{definition}\label{def:accept_lock}
Let $\Changed{\PP} \GETS \Message{Propose}{\v,\T,\PPproof{\Changed{\PP'}}}$ be a well-formed proposal in view $\v$. We say
that $\Changed{\PP'}$ is the \emph{preceding proposal} of $\Changed{\PP}$. For any two proposals $\PP_1$ and $\PP_2$, we say that $\PP_1$ \emph{precedes} $\PP_2$ if $\PP_1$ is the preceding proposal of $\PP_2$ or if there exists a
proposal $\Changed{\PP^{*}}$ such that $\PP_1$ precedes $\Changed{\PP^{*}}$ and $\Changed{\PP^{*}}$ is the preceding proposal of $\PP_2$. Let $\Precedes{\Changed{\PP}}$ be the set of all proposals that precede $\Changed{\PP}$. The \emph{depth} of proposal $\PP$ is defined by $\Depth{\PP} = \abs{\Precedes{\PP}}$.

We say that a replica \emph{records} $\Changed{\PP}$ if it determines that $\Changed{\PP}$ is well-formed (Line~\ref{alg:normal_case:recv} in Figure~\ref{alg:normal_case}) and say that it \emph{accepts} $\Changed{\PP}$ if it broadcasts \MName{Sync} messages with $\PPf{\Changed{\PP}}$.

We say that a replica \emph{conditionally prepares} $\Changed{\PP}$ if the replica received $\Changed{\PP}$ and, during view $\v$, \Changed{the replica receives $\n-\f$ concurring votes for $\Changed{\PP}$, i.e.} a set of messages $\{ \Message{Sync}{\v,\PPf{\Changed{\PP}_{\Replica[q]}},\AcceptedSet} \mid \Replica[q] \in \Replicas[q] \}$ with $\Changed{\PP}_{\Replica[q]} = \Changed{\PP}$ from a set $\Replicas[q] \subseteq \Replicas$ of $\abs{\Replicas[q]} = \n - \f$ replicas. We say that a replica \emph{conditionally commits} $\Changed{\PP}$ if, in a future view $\v[w] > \v$, the replica \emph{conditionally prepares} a proposal of the form $\Changed{\PP'} \GETS \Message{Propose}{\v[w], \T', \PPproof{\Changed{\PP}}}$ that extends $\Changed{\PP}$. We say that a replica \emph{locks} $\Changed{\PP}$ if $\Changed{\PP}$ is the highest proposal that it
\emph{conditionally commits}, denoted by $\PP{}_{\text{lock}}$. Also, we say that a replica \emph{commits} $\Changed{\PP}$ if, in a future view $\v[u] > \v$, the replica \emph{conditionally prepares} a proposal of the form
$\Changed{\PP''} \GETS \Message{Propose}{\v[u], \T'', \PPproof{\Changed{\PP'}}}$ that extends $\Changed{\PP'}$, with $u=w+1=\v+2$. \Changed{We say that two proposals are conflicting if the preceding proposals of these two proposals are disjoint.}

\end{definition}

\begin{figure}[t!]
    \begin{myprotocol}
    \STATE Let $\Primary{} \in \Replicas$ be the replica with $\v = \ID{\Primary{}} \bmod \n$ (the primary).
    \SPACE

    \FUCTION{Acceptable}{$\PP \GETS \Message{Propose}{\v,\T,\PPproof{\PP'}}$}
        \STATE Let $\PP_{\text{lock}} \GETS \Message{Propose}{\v_{\text{lock}},\T_{\text{lock}},\PPproof{\PP^{*}}}$ be the highest proposal that this replica conditionally committed.
        \RETURN $\Replica{}$ \emph{conditionally prepared} proposal $\PP'$ and either\\
\phantom{\textbf{return} }$\v'<v_{\text{lock}}$ or $\PP_{\text{lock}} \in (\{\PP'\} \union precedes(\PP'))$.
    \ENDFUNCTION
    \SPACE
    \FUCTION{HighestExtendable}{}
    \FOR{$\v$ from \textit{CurrentView} down to $0$}
        \IF{$\Primary{}$ \emph{conditionally prepared} proposal $\PP'$ of view $\v$}
            \IF{$\Primary{}$ has a valid $\PPproof{\PP'}$}
                \RETURN $\PP', \PPproof{\PP'}$.
            \ELSIF{$\Primary{}$ receives $\Message{Sync}{\v_i, \PPf{}_i, \AcceptedSet_i}$ from\\
                \phantom{\textbf{else if} }$\n-\f$ replicas $\Replica_i$ with $\PPf{\PP'} \in \AcceptedSet_i$}
                \RETURN $\PP', \PPf{\PP'}$.
            \ENDIF
        \ENDIF
    \ENDFOR
    \ENDFUNCTION
    \SPACE

    \TITLE{Primary role}{running at the primary $\Primary{}$ of view $\v$}
    \STATE $\PP, \text{cc} \GETS \Name{HighestExtendable}()$.
    \STATE Awaits receipt of a valid client request $\Cert{\T}{\Client}$.
    \STATE Broadcasts $\Message{Propose}{\v,\T, \text{cc}}$ to all replicas.\label{alg:normal_case:broadcast_proposal}
    \SPACE

    \TITLE{Backup role}{running at each replica $\Replica \in \Replicas$}
    \EVENT{$\Replica$ receives a \emph{well-formed} proposal $\PP$}
        \IF{$\Replica$ has not sent \MName{Sync} message in view $\v$
            and\\
\phantom{\textbf{if} }            \Name{Acceptable}($\PP$)}
            \STATE Broadcasts $\Message{Sync}{\v,\PPf{\PP}, \AcceptedSet}$ to all replicas.\label{alg:normal_case:recv}
        \ENDIF
    \ENDEVENT
    \SPACE

    \EVENT{$\Replica$ determines a failure in view $\v$}
        \STATE Broadcasts $\Message{Sync}{\v, \PPf{\varnothing}, \AcceptedSet}$ to all replicas.
    \ENDEVENT
    \SPACE

    \EVENT{$\Replica$ receives $\Message{Sync}{\v,\PPf{\PP}, \AcceptedSet}$ from\\
    \phantom{\textbf{event} }$\n-\f$ replicas}
        \STATE \emph{Conditionally prepares} $\PP$.
    \ENDEVENT
    \SPACE

    \EVENT{$\Replica$ receives $\Message{Sync}{\v',\PPf{}, \AcceptedSet}$ messages with\\
        \phantom{\textbf{event} }$\v' > v$ and $\PPf{\PP} \in \AcceptedSet$ from $\f+1$ replicas}
        \STATE \emph{Conditionally prepares} $\PP$.
    \ENDEVENT
    \SPACE

    \EVENT{$\Replica$ receives $\Message{Sync}{\v,\PPf{\PP}, \AcceptedSet}$ from\\
    \phantom{\textbf{event} }$\f+1$ replicas}
    \IF{$\Replica$ has not sent \MName{Sync} message in view $\v$}
    \STATE Broadcasts $\Message{Sync}{\v, \PPf{\PP}, \AcceptedSet}$ to all replicas.
    \ENDIF
    \IF{$\Replica$ does not know $\PP$}
    \STATE Send $\Message{Ask}{\v,\PPf{\PP}}$ to the $\f+1$ replicas.
    \ENDIF
    \ENDEVENT
    \SPACE

    \EVENT{$\Replica$ receives $\Message{Ask}{\v,\PPf{\PP}}$ from $\Replica{}'$ and\\
    \phantom{\textbf{event} }$\Replica$ has \emph{recorded} $\PP$}
        \STATE Send $\PP$ to $\Replica{}'$.
    \ENDEVENT
    \SPACE
    \end{myprotocol}
    \caption{\Changed{The replication protocol in a \SpotLess{} instance.}}\label{alg:normal_case}
\end{figure}

The proposal states \emph{conditionally committed} and \emph{committed} are analogous to the proposal states
\emph{prepared} and \emph{committed} in traditional non-chained protocols such as \PBFT{}~\cite{pbftj}. In
Figure~\ref{fig:normal_case}, we show how a proposal establishes the three states in the normal case. Using
this terminology, we can specify the safety guarantee that individual \SpotLess{} instances will maintain on the
system: \emph{no two conflicting proposals $\PP_{i}$ and $\PP_{j}$ can be both committed, each by a non-faulty replica},
which we will prove in Theorem~\ref{thm:safety}.

\Changed{To guarantee safety,} central to the design principle are the rules that non-faulty replicas follow when deciding whether to \emph{extend, accept}, or \emph{conditionally prepare} a proposal.

The primary $\Primary{}$ can construct a \emph{certificate} for proposal $\PP$ after $\Primary{}$
\emph{recorded} $\PP$ and received \Changed{$\n-\f$ \MName{Sync} messages with valid signatures for $\PP$, i.e.}
\[S = \{ \Message{Sync}{\v-1,\PPf{\PP{}_{\Replica[q]}},\AcceptedSet} \mid \Replica[q] \in \Replicas[q] \} \]
with $\PP{}_{\Replica[q]} = \PP$ and a valid signature from a set $\Replicas[q] \subseteq \Replicas$ of
$\abs{\Replicas[q]} = \n - \f$ replicas. The set $S$ will be used to construct the certificate. Even if $\Replica{}$ fails to receive sufficient \MName{Sync} messages
to \emph{conditionally prepare} $\PP$ in view $\v-1$, $\Replica{}$ will \emph{conditionally prepare} $\PP$ if
it receives a valid \emph{certificate} \Changed{$\PPproof{\PP}$}.

\Changed{Each \MName{Sync} message includes $\AcceptedSet$ that consists of the views and digests of the sender's
$\PP{}_{\text{lock}}$ and all \emph{conditionally prepared} proposals with a higher view than the view \Changed{$v_{\text{lock}}$} of proposal $\PP{}_{\text{lock}}$}:\Changed{
\[\AcceptedSet \GETS \{v_{\PP}, \digest{\PP} \mid \text{$\PP$ is conditionally prepared} \land v_{\text{lock}} \le v_{\PP}\}.\]}
$\Replica{}$ \emph{conditionally prepares} $\PP$ if it receives a set $S'$of \MName{Sync} messages from $\f+1$ 
replicas claiming to have \emph{conditionally prepared} $\PP$, \Changed{which implies at least one non-faulty replicas have \emph{conditionally prepared} $\PP$ after receiving $\n-\f$ concurring votes}, where
\[S' = \{ \Message{Sync}{w_{Q'},\PPf{\PP{}_{\Replica[q']}},\AcceptedSet_{Q'}} \mid \Replica[q'] \in \Replicas[q'] \} \]
with $w_{Q'}\ge \v$, $\PP \in \AcceptedSet_{Q'}$, and $\Replicas[q'] \subseteq \Replicas$ with $\abs{\Replicas[q']} = \f + 1$.

A non-faulty primary \Changed{of view $v$} considers a proposal $\PP'$ to be \emph{extendable} if either of the following conditions is met:
\begin{enumerate}
    \renewcommand{\labelenumi}{E\arabic{enumi}}
    \item \label{c1:proof-exist}$\Primary{}$ has a valid \emph{certificate} for $\PP'$; 
    \item \label{c2:nf-claim}\Changed{$\Primary{}$ has received a set of \MName{Sync} messages from $\n-\f$ replicas 
    that claim to have \emph{conditionally prepared} $\PP'$, i.e.
    \[\{ \Message{Sync}{w_{Q},\PPf{\PP{}_{\Replica[q]}},\AcceptedSet_{Q}} \mid \Replica[q] \in \Replicas[q] \}\] 
    with $w_{Q}$$<$$\v$ and
    $\PP'\in\AcceptedSet_{Q'}$ and $\Replicas[q] \subseteq \Replicas$ with $\abs{\Replicas[q]}=\n-\f$.}
\end{enumerate}

The primary $\Primary{}$ backtracks to earlier views to find the highest \emph{extendable} proposal \Changed{$\PP'$} and
then sets the preceding proposal to \Changed{$\PP'$}. \Changed{If $\PP'$ satisfies E\ref{c1:proof-exist}, then $\Primary{}$ broadcasts $\PP \GETS \Message{Propose}{\v,\T,\PPproof{\PP'}}$. Otherwise, $\Primary{}$ broadcasts $\PP \GETS \Message{Propose}{\v,\T,\PPf{\PP'}}$.}

When receiving a well-formed new proposal of the form
$\Changed{\PP} \GETS \Message{Propose}{\v,\T,\PPproof{\Changed{\PP'}}}$ or
$\Message{Propose}{\v,\T,\PPf{\Changed{\PP'}}}$, a replica $\Replica{}$ determines whether to \emph{accept} $\Changed{\PP}$ based on the following rules:

\begin{enumerate}
    \renewcommand{\labelenumi}{A\arabic{enumi}}
    \item \label{e0:validity} Validity Rule: $\Replica$ has \emph{conditionally prepared} $\Changed{\PP'}$.
    \item \label{e1:extending-lock} Safety Rule: $\Changed{\PP'}$ extends $\Replica$'s locked proposal $\PP_{\text{lock}}$, \Changed{i.e.} \\$\PP_{\text{lock}} \in (\{\Changed{\PP'}\} \union \Precedes{\Changed{\PP'}})$.
    \item \label{e2:higher-view} Liveness Rule: $\Changed{\PP'}$ has a higher view than $\PP_{\text{lock}}$.
\end{enumerate}

If \Changed{A\ref{e0:validity} holds} and either A\ref{e1:extending-lock} or A\ref{e2:higher-view} holds,
then $\Replica{}$ broadcasts $\Message{Sync}{\v, \PPf{\Changed{\PP}},\AcceptedSet}$. Otherwise, $\Replica{}$
keeps waiting for a proposal satisfying the acceptance requirement until its timer expires.

Due to unreliable communication or faulty behavior, \Changed{non-faulty replica $\Replica$ may fail to receive any \emph{acceptable} proposal from primary $\Primary{\v}$ but receive a set $M'$ consisting of $\f+1$ \MName{Sync} messages with the same $\PPf{\Changed{\PP}}$,}
formally, $\Replica{}$ receives \[M' = \{ \Message{Sync}{\v,\PPf{\PP{}_{\Replica[q^m]}},\AcceptedSet_{Q^{m}}} \mid \Replica[q^m] \in \Replicas[q^m] \} \]
with $\PP{}_{\Replica[q^m]} = {\Changed{\PP}}$ from a set $\Replicas[q^m] \subseteq \Replicas$ of $\abs{\Replicas[q^m]} = \f+1$ replicas.
For easier restoration of liveness, \SpotLess{} allows $\Replica{}$ to broadcast $\Message{Sync}{\v, \PPf{\Changed{\PP}},\AcceptedSet}$ if $\Replica{}$ considers \Changed{$\PP$ as} \emph{acceptable}.

In such a case, $\Replica{}$ is unaware of the full information of $\Changed{\PP}$ and needs to catch up. To do so,
$\Replica{}$ sends $\ma \GETS \Message{Ask}{\v,\PPf{\Changed{\PP}}}$ to the $\f+1$ replicas in $\Replicas[q^m]$. After a good replica $\Replica' \in \Replicas[q^m]$ receives $\ma$, the replica $\Replica{}'$ will forward $\Changed{\PP}$ to $\Replica{}$ if it has \emph{recorded} a well-formed $\Changed{\PP}$. To reduce the overhead of this mechanism in practical implementations, replicas can choose to first send \MName{Ask} messages to replicas they already trust (e.g., based on previous behavior).

\Changed{Based on the design principles above, we can prove the safety property of \SpotLess{} step by step:}

\begin{lemma}\label{lemma:qaccept}
    If a non-faulty replica $\Replica{}$ \emph{conditionally prepares} $\Changed{\PP} = \Message{Propose}{\v,\T,\PPproof{\Changed{\PP'}}}$, then
    for each proposal $\Changed{\PP^{*}} \in \Precedes{\Changed{\PP}}$ that precedes $\Changed{\PP}$, at least $\n-2\f\geq \f+1$ non-faulty replicas have \emph{conditionally prepared} $\Changed{\PP^{*}}$ and sent \MName{Sync} messages with $\Changed{\PP^{*}} \in \AcceptedSet$.
\end{lemma}

\begin{proof}
    Assume that $\Replica{}$ is the first non-faulty replica that \emph{conditionally prepares} $\PP$. The only way for $\Replica{}$ to \emph{conditionally prepare} $\PP$ is to receive $\n-\f$ \MName{Sync} messages of the form $\Message{Sync}{\v,\PPf{\PP}, \AcceptedSet}$, of which $\n-2\f\geq \f+1$ are from non-faulty replicas. Hence, at least $\n-2\f$ non-faulty replicas must have \emph{conditionally prepared} the preceding proposal $\PP'$ of $\PP$, as non-faulty replica only accept $\PP$ if they \emph{conditionally prepared} $\PP'$. From the definition of
    $\AcceptedSet$ we know that these $\n-2\f$ non-faulty replicas must have send \MName{Sync} messages with
    $\PP' \in \AcceptedSet$ after they \emph{conditionally prepare} $\PP'$. Next, we can bootstrap this argument to show that the above statement holds for each proposal $\PP^{*} \in \Precedes{\PP}$.
\end{proof}

Using Lemma~\ref{lemma:qaccept}, we are able to prove safety:

\begin{theorem}\label{thm:safety}
    No two non-faulty replicas can commit conflicting proposals $\PP_{i}$ and $\PP_{j}$.
\end{theorem}

\begin{proof}
    We shall prove this theorem by contradiction. Assume that both $\PP_i$ and $\PP_j$ are \emph{committed} by (possibly distinct) non-faulty replicas and that $\PP_i$ and $\PP_j$ are conflicting. Without loss of generality, we can assume that proposal $\PP_i$ is from view $v_i$, proposal $\PP_j$ is from view $v_j$, and $\Depth{\PP_j} \geq \Depth{\PP_i}$. As $\PP_i$ is \emph{committed}, there must exist proposals $\PP_{i+1}$ and $\PP_{i+2}$ of the form 
    \begin{align*}
    \PP_{i+1} &= \Message{Propose}{\v_{i}+1,\T_{i+1},\PPproof{\PP_{i}}};\text{ and}\\
    \PP_{i+2} &= \Message{Propose}{\v_{i}+2,\T_{i+2},\PPproof{\PP_{i+1}}}
    \end{align*}
    that have been \emph{conditionally prepared} by non-faulty replicas. Likewise, as $\PP_j$ is \emph{committed}, there must exists proposals $\PP_{j+1}$ and $\PP_{j+2}$ that have been \emph{conditionally prepared} by non-faulty replicas.
    
    Since $\PP_{j}$ conflicts $\PP_{i}$ and $\Depth{\PP_j} \geq \Depth{\PP_i}$, there must exists a depth $d$ and proposals $\PP_{d,i} \in \{ \PP_i \} \union \Precedes{\PP_i}$ and $\PP_{d,j} \in \{ \PP_j \} \union \Precedes{\PP_j}$ such that $\Precedes{\PP_{d,i}} = \Precedes{\PP_{d,j}}$ and $\PP_{d,i} \neq \PP_{d,j}$ (hence, the first proposals that precede $\PP_i$ and $\PP_j$, respectively, that are in conflict). As both $\PP_i$ and $\PP_j$ are \emph{committed}, some non-faulty replicas must have \emph{conditionally prepared} $\PP_{d,i}$ and $\PP_{d,j}$.
    
    From Lemma~\ref{lemma:qaccept} and the fact that $\PP_{i+2}$ and $\PP_{j+2}$ are \emph{conditionally prepared} by a non-faulty replica, we conclude that $\n-2\f\geq \f+1$ non-faulty replicas \emph{conditionally prepared} all proposals in $\Precedes{\PP_{i+2}}$ and in $\Precedes{\PP_{j+2}}$. Let $\PP' \in \Precedes{\PP_{j+2}}$ be the proposal with $\Depth{\PP_{i+1}} = \Depth{\PP'}$. By construction, $\PP_d \in \Precedes{\PP'}$. hence, $\PP'$ and $\PP_{i+1}$ are in conflict and are both \emph{conditionally prepared} by $\n-2\f\geq \f+1$ non-faulty replicas. Hence, by a similar quorum argument as used to prove Theorem~\ref{thm:unique}, we conclude there must be at least one non-faulty replica that \emph{conditionally prepared} both $\PP'$ and $\PP_{i+1}$. This non-faulty replica must have also locked on either $\PP_i$ (when it accepted $\PP_{i+1}$) or the proposal $\PP''$ preceding $\PP'$ (when it accepted $\PP'$). By construction, $\PP'' \neq \PP'$. Hence, one of these accept steps would violate the safety rule A2, a contradiction. 
\end{proof}

Theorem~\ref{thm:safety} proves \SpotLess{} can ensure safety if we have \emph{three-consecutive-view} requirement for 
committing a proposal. Next, we illustrate the necessity of the \emph{three-consecutive-view} requirement (over a two-consecutive-view requirement) for \SpotLess{}.

\begin{example}
Assume a system with $\n = 3\f+1$ replicas and that all replicas have \emph{conditionally prepared} \Changed{$\PP_0$}. We describe events that happen in the next six views to illustrate how two conflicting proposals can get committed if we relax the \emph{three-consecutive-view} requirement.

    \begin{enumerate}
        \item Primary $\Primary{1}$ broadcasts proposal $\PP_1$ extending $\PP_0$. All replicas \emph{accept} $\PP_1$ and then \emph{conditionally prepare} $\PP_1$.
        \item Primary $\Primary{2}$ broadcasts proposal $\PP_2$ extending $\PP_0$. All replicas \emph{accept} $\PP_2$ and then \emph{conditionally prepare} $\PP_2$. 
        \item \Changed{Faulty primary $\Primary{3}$ sends proposal $\PP_3$ extending $\PP_2$ to $\f+1$ non-faulty replicas including $\Replica_0$ and sends proposal $\PP_3'$ extending $\PP_2$ to the other $\f$ non-faulty replicas. All $\f$ faulty replicas only send \MName{Sync} messages with $\PPf{\PP_3}$ to $\Replica_{0}$, due to which only $\Replica_{0}$ \emph{conditionally prepares} $\PP_3$, while other non-faulty replicas cannot \emph{conditionally prepare} it.
        \item $\Primary{4}$ broadcasts proposal $\PP_4$ extending $\PP_1$. All non-faulty replicas except $\Replica_{0}$ \emph{conditionally prepare} $\PP_4$.
        \item Faulty primary $\Primary{5}$ sends proposal $\PP_5$ extending $\PP_4$ to $\f+1$ non-faulty replicas including $\Replica_1$ and not including $\Replica_0$ and sends proposal $\PP_5$ extending $\PP_4$ to the other $\f$ non-faulty replicas. All $\f$ faulty replicas only send \emph{Sync} messages with $\PPf{\PP_5}$ to $\Replica_{1}$, due to which  only $\Replica_{1}$ \emph{conditionally prepares} $\PP_5$ and \emph{commits} $\PP_1$, while others cannot \emph{conditionally prepares} $\PP_5$.
        \item $\Primary{6}$ broadcasts proposal $\PP_6$ extending $\PP_3$ and all replicas except $\Replica_{1}$ broadcast \MName{Sync} messages with $\PPf{\PP_6}$. Then, all non-faulty replicas except $\Replica_{1}$ \emph{conditionally prepare} $\PP_6$ and \emph{commit} $\PP_2$ that conflicts with $\PP_1$, which was \emph{committed} by $\Replica_{1}$.}
 \end{enumerate}

    In this setting, using a two-consecutive-view requirement for committing, $\PP_1$ and $\PP_2$ are conflicting proposals but are both committed by non-faulty replicas.
\end{example}

\subsection{\Changed{Bootstraping Liveness with Rapid View Synchronization}}\label{ss:rvs}

\emph{Rapid View Synchronization}  (\RVS{}) bootstraps the guarantees provided by normal-case replication
toward providing consensus. \RVS{} does so by dealing with asynchronous communication and by strengthening
the guarantees on proposals of preceding views. \Changed{In specific, the main services provided by \RVS{} are a best-effort and quick \emph{view synchronization} to assure that replicas end up in the same views whenever communicaiton is sufficiently reliable and \emph{low-cost state recovery} to enable cheap primary rotation to deal with failures of previous primaries.}

To enable \emph{Rapid View Synchronization}, for each view $\v$, a replica must go through three
states one by one:
\begin{enumerate}
    \renewcommand{\labelenumi}{ST\arabic{enumi}}
    \item\label{state:recording} Recording: waiting for a well-formed $\PP$ that satisfies A\ref{e0:validity} and
    either \Changed{A\ref{e1:extending-lock} or A\ref{e2:higher-view}} until state timer $\timer{R}$ expires;
    \item\label{state:syncing} Syncing: waiting for a set of \MName{Sync} messages with view $\v$ from a set
    $\Replicas[q] \subseteq \Replicas$ with $\abs{\Replicas[q]} = \n-\f$ replicas;
    \item\label{state:certifying} Certifying: waiting for a set of messages \[S = \{ \Message{Sync}{\v,\PPf{\PP{}_{\Replica[q^A]}},\AcceptedSet} \mid \Replica[q^A] \in \Replicas[q^A] \}\]
    with the same claimed proposal $\PP{}_{\Replica[q^A]}$ from a set $\Replicas[q^A] \subseteq \Replicas$ of $\abs{\Replicas[q^A]} = \n-\f$ replicas until timer state $\timer{A}$ expires.
\end{enumerate}
Note that there is no timer for Syncing (ST\ref{state:syncing}) and receiving sufficient \MName{Sync} messages is the only
way to proceed to Certifying (ST\ref{state:certifying}) of the same view. Some replicas may fall behind due to unreliable communication, failing to receive sufficient \MName{Sync} messages
while other replicas have reached higher views. To quickly synchronize views, a replica $\Replica{}$ in view $\v$
is allowed to proceed to Syncing (ST\ref{state:syncing}) of view $w$ directly after receiving a set of messages $D$ with
views higher than or equal to $w$: \[D = \{ \Message{Sync}{\v',\PPf{\PP{}_{\Replica[q^d]}},\AcceptedSet_{Q^{d}}} \mid \Replica[q^d] \in \Replicas[q^d] \} \]
with $\v'\ge w>\v$ from a set $\Replicas[q^d] \subseteq \Replicas$ of $\abs{\Replicas[q^d]} = \f + 1$ replicas. \Changed{Receiving such $\f+1$ messages implies that one non-faulty replica has moved to view $w$ after receiving $\n-\f$ \MName{Sync} messages of view $w-1$, then $\Replica{}$ can skip to view $w$ directly knowing that at least majority of non-faulty replicas have observed the higher view $w-1$}. To catch up, replica $\Replica{}$ broadcasts message $S_u = \Message{Sync}{u,\PPf{\varnothing},\AcceptedSet, \Upsilon}$ for each view $u, \v \le u \le w$, in which $\Upsilon$ is a flag that asks replicas that receive $S_u$ to retransmit the \MName{Sync} messages to $\Replica{}$ they broadcast in view $u$. \Changed{With such a design, in \SpotLess{}, as long as network remains synchronous, replicas falling behind are capable of catching up actively and immediately, while in previous \emph{rotational} work such as \HS{}, \emph{view synchronization} is assumed by relying on the black-box Pacemaker}. In Figure~\ref{alg:rvs}, we present the pseudo-code of the Rapid View Synchronization part of \SpotLess{}. Due to Lemma~\ref{lemma:qaccept}, \Changed{we have the following:}

\begin{lemma}\label{lemma:accept_precede}
    Assume reliable communication. If replica $\Replica{}$ \emph{conditionally prepares} proposal $\PP$ by receiving $\f+1$
    \MName{Sync} messages with $\PPf{\PP} \in \AcceptedSet$, then eventually $\Replica{}$ will \emph{record} and
    \emph{conditionally prepare} all proposals in $\Precedes{\PP}$.
\end{lemma}

\begin{proof}
    Let  $\PP' \in \Precedes{\PP}$. Due to Lemma~\ref{lemma:qaccept}, at least $\n-2\f \geq \f+1$ non-faulty replicas have sent \MName{Sync} messages that claim $\PP'$. As such, at-least $\f+1$ non-faulty replicas will reply to the message $S_u = \Message{Sync}{u,\PPf{\varnothing},\AcceptedSet, \Upsilon}$ sent by replica $\Replica$. Hence, eventually, $\Replica{}$ will receive $\f+1$ corresponding \MName{Sync} messages with $\PPf{\PP'}$, due to which $\Replica$ will \emph{record} and \emph{conditionally prepare} $\PP'$.
\end{proof}

Using Lemma~\ref{lemma:accept_precede}, we can prove:

\begin{theorem}\label{thm:cpall}
Let \Changed{$\PP_h$} be the highest proposal that any replica conditionally committed. All non-faulty replicas will eventually \emph{record} and \emph{conditionally prepare} all proposals in $\Precedes{\Changed{\PP_h}}$, \Changed{when communication becomes synchronous for sufficiently long.}
\end{theorem}

\begin{proof} 
    Since $\PP_h$ is \emph{locked} by a non-faulty replica, at least $\n-2\f\geq\f+1$ non-faulty replicas have
    \emph{conditionally prepared} $\PP_h$ (Lemma~\ref{lemma:qaccept}). These $\n-2\f$ non-faulty replicas either locked $\PP_h$ or some proposal with a lower view. Hence, all these replicas will broadcast \MName{Sync} messages with $\PPf{\PP_h} \in \AcceptedSet$ to convince other replicas to \emph{conditionally prepare} $\PP_h$. Hence, when communication becomes synchronous, Lemma~\ref{lemma:accept_precede} completes this proof.
\end{proof}

From Theorem~\ref{thm:cpall} we know that all non-faulty replicas will learn the same 
\emph{conditionally committed} chain. However, the replicas may not execute several proposals on the chain until
they learn full information of the proposal via the \MName{Ask}-recovery mechanism detailed in Section~\ref{ss:rules}.

\begin{figure}[t]
    \begin{myprotocol}

    \TITLE{Backup role}{running at each replica $\Replica \in \Replicas$ in view $\v$}
    \EVENT{$\Replica$ enters view $\v$}
        \STATE $\text{state} \GETS \textit{recording}$ (ST\ref{state:recording}).
        \STATE Set timer $\timer{R}$.
    \ENDEVENT
    \SPACE

    \EVENT{$\Replica$ receives an acceptable proposal $\PP$ \textbf{or} $\timer{R}$ expires}
        \STATE Broadcast $\Message{Sync}{\v, \PPf{\PP}, \AcceptedSet}$.
        \STATE $\text{state} \GETS \textit{Syncing}$ (ST\ref{state:syncing}).
    \ENDEVENT
    \SPACE

    \EVENT{$\Replica$ receives $\n-\f$ \MName{Sync} messages of view $\v$}
        \STATE $\text{state} \GETS \textit{Certifying}$ (ST\ref{state:certifying}).
        \STATE Sets timer $\timer{A}$.
    \ENDEVENT
    \SPACE

    \EVENT{$\Replica$ receives $\n-\f$ $\Message{Sync}{\v, \PPf{\PP}, \AcceptedSet}$ of the same $\PP$ \textbf{or} $\timer{A}$ expires}
        \STATE Enters view $\v+1$.
    \ENDEVENT
    \SPACE

    \EVENT{$\Replica$ receives \Changed{$\f+1$} \MName{Sync} messages \Changed{with views higher than or equal to $w$, $w > \v$}}
        \STATE Let $v$ be the current view of $\Replica$.
        \FOR{each view $u, \v \le u \le w$}
        \STATE Broadcasts $\Message{Sync}{u,\PPf{\varnothing},\AcceptedSet, \Upsilon}$.
        \ENDFOR
    \ENDEVENT
    \SPACE

    \end{myprotocol}
    \caption{Rapid View Synchronization in instance.}\label{alg:rvs}
\end{figure}

\subsection{\Changed{Mechanism Guaranteeing Liveness}}\label{ss:timer}

In some cases, replica $\Replica{}$ cannot make any progress unless it receives some specific messages
from other replicas:

\begin{enumerate}
    \item $\Replica{}$ cannot switch from Syncing (ST\ref{state:syncing}) to Certifying (ST\ref{state:certifying}) unless
    it receives $\n-\f$ \MName{Sync} messages of its current view.
    \item $\Replica{}$ cannot catch up to learn a path from a \emph{conditionally prepared} proposal
    to the genesis proposal unless at least $\f+1$ other replicas reply to its \MName{Sync} messages
    with flag $\Upsilon$\Changed{, which requires the receivers to retransmit the \MName{Sync} messages that they broadcast before.}
    \item $\Replica{}$ cannot \emph{record} a proposal it did not receive from the primary, unless any replica
    replies to its \MName{Ask} message by forwarding the corresponding \MName{Propose} message.
\end{enumerate}

However, due to unreliable communication, $\Replica{}$ may fail to receive messages, e.g., replies to \MName{Sync} messages with \Changed{flag $\Upsilon$} or replies to \MName{Ask} messages. To deal with this case, $\Replica{}$ will periodically retransmit the messages until it receives the necessary replies.

In an asynchronous environment, one cannot reliably distinguish between communication failure (e.g., due to long and unpredictable message delays) and replica failure. Hence, consensus protocols such as \HS{}~\cite{hotstuff} and many others~\cite{thetabyz,lewis} simply assume to be operating after a \emph{Global Synchronization Time}, at which point all communication is bound by some message delay such that all replicas can always reliably determine in which view they operate~\cite{hotstuff}. Such a design is inflexible in the presence of true asynchronous communication, however.

Instead, \SpotLess{} instances use our \Changed{\emph{Rapid View Synchronization}} mechanism to allow replicas to figure out in which view they should operate. To adapt to fluctuations in message delays, \SpotLess{} will adjust the timeout interval used by individual replicas to detect replica failures. As message delays in typical deployments do not often change drastically, we choose to \emph{not} use a traditional exponential backoff mechanism~\cite{book}, but instead to adjust the timeout interval of replicas $\Replica{}$ in a more moderate way. For consecutive timeouts of the same timer in consecutive views, we only increase the timeout interval by a constant $\varepsilon$ (after each consecutive view). If a replica receives an expected message for which the timeout interval was $\Delta$ \emph{before} $0.5\Delta$, then the replica reduces the timeout by half. We have the following technical result.

\begin{lemma}\label{lemma:all_vote}
    Let $\v$ be the highest view reached by a non-faulty replica after communication enters a period of synchronous communication. Non-faulty primary $\Primary{w}$ with $\v + 2 \le w$ is capable of finding a proposal \Changed{$\PP'$ such that all non-faulty replicas will accept a proposal $\PP$ extending from $\PP'$, $\PP \GETS \Message{Propose}{w,\T,\PPf{\PP'}}$}.
\end{lemma}

\begin{proof}
    Let $\PP'$ be the highest proposal \emph{locked} by any non-faulty replicas after view $w-1$ (in the worst case, $\PP'$ is the genesis proposal).
    As $\PP'$ was locked by a non-faulty replica, at least $\n-2\f\geq \f+1$ non-faulty replicas have \emph{conditionally prepared} $\PP'$ before view $w-1$.
    For any replica $\Replica{}$ among these $\n-2\f$ replicas, either $\Replica{}$ \emph{locked} on $\PP'$ or $\PP'$ is from a newer view than the locked proposal of $\Replica$. 
    
    Hence, in view $w-1$, the \MName{Sync} messages of the $\n-2\f\geq \f+1$ non-faulty replicas would include $\PP'$ in $\AcceptedSet$. As such, all non-faulty replicas would \emph{conditionally prepare} $\PP'$ and inform other replicas, and primary $\Primary{w}$ would learn that at least $\n-2\f$ non-faulty replicas have \emph{conditionally prepared} $\PP'$ before entering view $w$. As such, primary $\Primary{w}$ can propose a new proposal $\PP \GETS \Message{Propose}{w,\T,\PPf{\PP^{*}}}$, with $\PP^{*} = \PP'$ or with $\PP^{*}$ being some proposal with a higher view than $\PP'$ (if primary $\Primary{w}$ knows such an \emph{extendable} higher proposal). All non-faulty replicas will consider $\PP$ as \emph{acceptable} since their \emph{locked} proposal is either $\PP'$ or a proposal from an lower view.
\end{proof}

\Changed{As a direct consequence, we have the following corollary.}

\begin{corollary}\label{cor:all_accept}
Let $\v$ be the highest view reached by a non-faulty replica after communication enters a period of synchronous communication. If all non-faulty replicas have internal timers that are higher than the current maximum message delay, then the proposals of any non-faulty primary $\Primary{w}$ during view $w$, $w \geq \v+2$, will be \emph{conditionally prepared} after view $w$ by all non-faulty replicas.
\end{corollary}

\begin{proof}
    We denote to be $t$ the time that $\Primary{w}$ enters Certifying (\ref{state:certifying}) of view $w$.
    From Lemma~\ref{lemma:all_vote} we know that all non-faulty replicas will \emph{accept}.
    By $t$, at least $\f+1$ non-faulty replicas have sent $\Message{Sync}{w,\PPf{m},\AcceptedSet}$.
    By $t+\delta$, non-faulty replicas falling behind will have received at least $\f+1$ such
    \MName{Sync} messages, entered view $w$ and voted for $m$. By $t+2\delta$, all non-faulty
    replicas will have received at least $\n-\f$ such \MName{Sync} messages and \emph{conditionally prepare} $m$.
\end{proof}

Corollary~\ref{cor:all_accept} is at the basis of proving termination: consensus decisions are eventually made when communication is synchronous.

\begin{theorem}\label{theorem:liveness}
All non-faulty replicas will eventually commit new proposals after communication enters a sufficiently-long period of synchronous communication.
\end{theorem}

\begin{proof}
    First, we note that all non-faulty replicas will eventually satisfy the internal timer requirements, as they will increase their internal timer until they can successfully participate in the replication protocol. Hence, by Corollary~\ref{cor:all_accept} a non-faulty primary $\Primary{w}$ of view $w$ will eventually be able to propose a message $\PP$ that will be conditionally prepared. As there are $\n > 3\f$ replicas, after view $w$, there will eventually be a sequence of three consecutive non-faulty primaries. These  primaries will be able to make proposals satisfying the conditions of Definition~\ref{def:accept_lock} to assure a proposal will get committed.
\end{proof}

\section{Concurrent Consensus}\label{sec:concurrent}

The main benefit of chained consensus, as used by \SpotLess{} and \HS{}~\cite{hotstuff}, is that a single proposal represents the entire chain of preceding proposals. This greatly reduces the message complexity of view-changes when compared to traditional non-chained consensus protocols such as \PBFT{}~\cite{pbftj}.

Unfortunately, chained consensus requires that consecutive consensus decisions are made one-at-a-time, thereby preventing the usage of \emph{out-of-order processing} to maximize throughput. This makes \HS{} and individual chained consensus instances of \SpotLess{} significantly slower than traditional consensus protocols such as \PBFT{} in practical deployments: primaries in \PBFT{} can use \emph{out-of-order processing} to propose client requests for future views while waiting on the current consensus round to finish, thereby maximizing the utilization of the network bandwidth available at the primary independent of any message delays (which dominate the time it takes to finish a single consensus round).

As an alternative to out-of-order processing, \SpotLess{} will adopt concurrent consensus~\cite{rcc,mirbft}. By running multiple concurrent instances, \SpotLess{} is able to effectively utilize all network bandwidth and computational resources available: when one \SpotLess{} instance is waiting for a proposal to be processed (e.g., waiting for \MName{Sync} messages), other \SpotLess{} instances use the available network bandwidth and computational resources to propose additional requests.

\subsection{Concurrent Instances in SpotLess}

In \SpotLess{}, the system runs $\m$, $1 \le \m \le \n$, \SpotLess{} instances 
concurrently. Instances do not interfere with each other. Each instance 
independently deals with any malicious behavior. To enforce that each instance is coordinated by a distinct primary, the primary 
of $\Instance{i}$ in view $v$ is predetermined: $\ID{\Primary{i,v}} = (i+v) \bmod \n$. 
Figure~\ref{fig:primary_rotation} shows how \SpotLess{} assigns and rotates primary in 
each \SpotLess{} instance. As all instances rotate over all primaries, all instances are equally affected by malicious behavior. When given the choice, non-faulty replicas will prioritize instances that are in older views over other instances. Due to primary rotation and this instance prioritization, the view of all instances will remain roughly-in-sync.
\begin{figure}[t]
    \begin{tikzpicture}[]
        
        \node[font=\bfseries,align=center] at (2.6,  3.2) {$\Replica_0$};
        \node[font=\bfseries,align=center] at (4.2,  3.2) {$\Replica_1$};
        \node[font=\bfseries,align=center] at (5.8,  3.2) {$\Replica_2$};
        \node[font=\bfseries,align=center] at (7.4,  3.2) {$\Replica_3$};

        \node[font=\bfseries,align=center] at (1,  2.4) {\emph{View 0}};
        \node[font=\bfseries,align=center] at (1,  1.6) {\emph{View 1}};
        \node[font=\bfseries,align=center] at (1,  0.8) {\emph{View 2}};

        \draw[line width=1pt, double distance=3pt,
        arrows = {-Latex[length=0pt 2 0]}] (2.6, 2.8) -- (2.6, 0);
        \draw[line width=1pt, double distance=3pt,
        arrows = {-Latex[length=0pt 2 0]}] (4.2, 2.8) -- (4.2, 0);
        \draw[line width=1pt, double distance=3pt,
        arrows = {-Latex[length=0pt 2 0]}] (5.8, 2.8) -- (5.8, 0);
        \draw[line width=1pt, double distance=3pt,
        arrows = {-Latex[length=0pt 2 0]}] (7.4, 2.8) -- (7.4, 0);
                
                \draw[thick,red!30!white!90,fill=white] (2.6,2.4) circle (0.30cm);
                \draw[thick,red!30!white!90,fill=white] (4.2,1.6) circle (0.30cm);
                \draw[thick,red!30!white!90,fill=white] (5.8,0.8) circle (0.30cm);
                \draw[thick,yellow!,fill=white] (4.2,2.4) circle (0.30cm);
                \draw[thick,yellow!,fill=white] (5.8,1.6) circle (0.30cm);
                \draw[thick,yellow!,fill=white] (7.4,0.8) circle (0.30cm);
                \draw[thick,black!30!black!30,fill=white] (5.8,2.4) circle (0.30cm);
                \draw[thick,black!30!black!30,fill=white] (7.4,1.6) circle (0.30cm);
                \draw[thick,black!30!black!30,fill=white] (2.6,0.8) circle (0.30cm);
                \draw[thick,black!10!green!30,fill=white] (7.4,2.4) circle (0.30cm);
                \draw[thick,black!10!green!30,fill=white] (2.6,1.6) circle (0.30cm);
                \draw[thick,black!10!green!30,fill=white] (4.2,0.8) circle (0.30cm);

                \node[font=\bfseries,align=center] at (2.6, 2.4) {$\PP_{0,0}$};
                \node[font=\bfseries,align=center] at (4.2, 2.4)   {$\PP_{1,0}$};
                \node[font=\bfseries,align=center] at (5.8, 2.4) {$\PP_{2,0}$};
                \node[font=\bfseries,align=center] at (7.4, 2.4)   {$\PP_{3,0}$};
                \node[font=\bfseries,align=center] at (2.6, 1.6) {$\PP_{3,1}$};
                \node[font=\bfseries,align=center] at (4.2, 1.6)   {$\PP_{0,1}$};
                \node[font=\bfseries,align=center] at (5.8, 1.6) {$\PP_{1,1}$};
                \node[font=\bfseries,align=center] at (7.4, 1.6)   {$\PP_{2,1}$};
                \node[font=\bfseries,align=center] at (2.6, 0.8) {$\PP_{2,2}$};
                \node[font=\bfseries,align=center] at (4.2, 0.8)   {$\PP_{3,2}$};
                \node[font=\bfseries,align=center] at (5.8, 0.8) {$\PP_{0,2}$};
                \node[font=\bfseries,align=center] at (7.4, 0.8)   {$\PP_{1,2}$};
    \end{tikzpicture}
    \caption{Primary rotation in \SpotLess{} with four replicas and four instances. \Changed{
    The circles in the form of $\PP_{i, v}$ on the arrow of $\Replica_{r}$ represent that $\Replica_{r}$ is the primary 
    of instance $\Instance{i}$ in view $v$, where $r = (i + v) \bmod \n$.}}\label{fig:primary_rotation}
\end{figure}
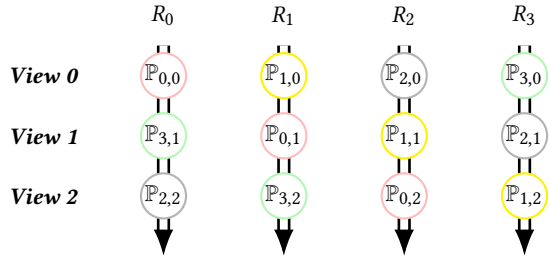

Each \SpotLess{} instance determines a local order of proposals. All committed proposals 
on the chains are totally ordered and then executed. We order the committed proposals 
among different instances by their view and instance identifier. We order all proposals from low 
view to high view and from instance $\Instance{0}$ to instance $\Instance{\m-1}$. Figure~\ref{fig:total_order} 
shows the \emph{total ordering} in \SpotLess{}. Finally, each replica \emph{executes} the committed proposals in order and 
informs the clients of the outcome of their requested transactions.

As individual \SpotLess{} instances provide consensus and combining consensus decisions of instances is deterministic, we conclude
\begin{theorem}\label{thm:cterm}
All instances of \SpotLess{} will eventually commit new proposals after communication enters a sufficiently-long period of synchronous communication. 
\end{theorem}

\subsection{Benefits of Concurrent Processing}

Theoretically, if we run $\m$, $1\le \m\le \n$ instances concurrently, \SpotLess{} could achieve $\m$ times the throughput of a single instance. As the number of replicas $\n$ and instances $\m$ scales, the throughput of \SpotLess{} will keep  growing until the throughput reaches resource bottlenecks (e.g., limited computational power or network bandwidth). Scaling beyond these resource bottlenecks, the throughput of \SpotLess{} decreases as the system scales further due to the added communication cost to reach consensus among more replicas.

Next, we shall theoretically model the benefits of \emph{concurrent consensus} in \SpotLess{}. Assume the message delay is $\Delta$ and the network bandwidth is $B$, and that proposal includes a \emph{batch} of $\beta$ individual transactions. Hence, the best-case throughput of individual \SpotLess{} instances with $\n$ replicas 
is
\begin{align*}\TP{\SpotLess_1} &= \frac{\beta}{t_\text{primary} + 2\Delta};& t_\text{primary} &= \frac{S_\text{primary}}{B} \end{align*}
In the above, $S_\text{primary}$ is the size of all $\n-\f$ \MName{Sync} messages the primary receives and of all $\n-\f$ \MName{Propose} messages the primary sends and  $t_\text{primary}$ is the time the primary is busy sending and receiving. 

If network bandwidth is not the bottleneck, then the maximum performance of $\n$ \SpotLess{} instances is $\n T_{\SpotLess_1}$. If we run $\n$ \SpotLess{} instances concurrently, one for each replica, the communication complexity per instance is $O(\n^2)$ in each round. Thus, as the system reaches resource bottlenecks, adding more instances will cause no gains in throughput. If we assume that network bandwidth is the bottleneck, the best-case throughput of \SpotLess{} with $\n$ instances and $\n$ replicas is upper-bounded by
\[
            \TP{\SpotLess_{\text{bw}}} = \frac{\n B\beta}{S_\text{primary}+(\n-1)S_\text{backup}}
\]
In the above, $S_\text{backup}$ is the size of all messages a replica sends and receives as part of the backup role in one view, and $S_\text{primary}+(\n-1)S_\text{backup}$ is the sum of all bandwidth usage by a single replica during $\n$ instances of a single view. Due to the chained design of \SpotLess{}, \SpotLess{} instances end up sending fewer messages per consensus decision than \PBFT{}. Consequently, when scaling up to many replicas (and many instances), the best-case throughput of \SpotLess{} will be higher than concurrent consensus based on \PBFT{} (e.g., \RCC{}~\cite{rcc}).


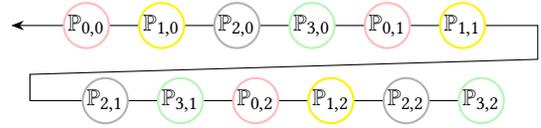
\begin{figure}[t]
    \begin{tikzpicture}
    
        \draw[<-] (-0.5, 0) -- (6.5, 0) -- (6.5, -0.45) -- (-0.25, -0.65) -- (-0.25, -1) -- (5.5, -1);
    
        \draw[thick,red!30!white!90,fill=white] (0.5, 0) circle (0.30cm);
        \node[font=\bfseries,align=center] at (0.5, 0) {$\PP_{0,0}$};
        \draw[thick,yellow!,fill=white] (1.5, 0) circle (0.30cm);
        \node[font=\bfseries,align=center] at (1.5, 0)   {$\PP_{1,0}$};
        \draw[thick,black!30!black!30,fill=white] (2.5, 0) circle (0.30cm);
        \node[font=\bfseries,align=center] at (2.5, 0) {$\PP_{2,0}$};
        \draw[thick,black!10!green!30,fill=white] (3.5, 0) circle (0.30cm);
        \node[font=\bfseries,align=center] at (3.5, 0)   {$\PP_{3,0}$};

        \draw[thick,red!30!white!90,fill=white] (4.5, 0) circle (0.30cm);
        \node[font=\bfseries,align=center] at (4.5, 0) {\Changed{$\PP_{0,1}$}};
        \draw[thick,yellow!,fill=white] (5.5, 0) circle (0.30cm);
        \node[font=\bfseries,align=center] at (5.5, 0)   {\Changed{$\PP_{1,1}$}};
        \draw[thick,black!30!black!30,fill=white] (0.75, -1) circle (0.30cm);
        \node[font=\bfseries,align=center] at (0.75, -1) {\Changed{$\PP_{2,1}$}};
        \draw[thick,black!10!green!30,fill=white] (1.75, -1) circle (0.30cm);
        \node[font=\bfseries,align=center] at (1.75, -1)   {\Changed{$\PP_{3,1}$}};
        
        \draw[thick,red!30!white!90,fill=white] (2.75, -1) circle (0.30cm);
        \node[font=\bfseries,align=center] at (2.75, -1) {$\PP_{0,2}$};
        \draw[thick,yellow!,fill=white] (3.75, -1) circle (0.30cm);
        \node[font=\bfseries,align=center] at (3.75, -1)   {$\PP_{1,2}$};
        \draw[thick,black!30!black!30,fill=white] (4.75, -1) circle (0.30cm);
        \node[font=\bfseries,align=center] at (4.75, -1) {$\PP_{2,2}$};
        \draw[thick,black!10!green!30,fill=white] (5.75, -1) circle (0.30cm);
        \node[font=\bfseries,align=center] at (5.75, -1)   {$\PP_{3,2}$};
    \end{tikzpicture}
    \caption{Total ordering of the twelve proposals of Figure~\ref{fig:primary_rotation} made among four concurrent instances in three consecutive views.}\label{fig:total_order}
\end{figure}

\section{Clients and Transaction Execution}\label{sec:client}

Lastly, we discuss how \SpotLess{} provides service to clients. Observe that to maximize throughput, it is beneficial to assure that no two instances propose the same transactions. In addition, it is important to balance all requests over all instances. Furthermore, we must also consider malicious primary behavior aimed at refusing service to some non-faulty clients.

Unlike systems such as \RCC{}~\cite{rcc}, which initially assigns every client to a single primary (that coordinates a single instance), \SpotLess{} assigns every client request to a single instance based on its digest: instance $\Instance{i}$, $1 \leq i \leq \m$, can only propose transactions with digest $d$ such that $(i - 1) = d \bmod \m$. Hence, multiple requests by the same client will be handled by different instances.  Since we adopt a cryptographically strong hash algorithm to compute digests, this allocation of transactions to instances also ensures load balance among instances.

When a non-faulty replica commits a transaction, it will execute that transaction in an order consistent with the total ordering of all transactions (across all instances). After executing a transaction, the replica informs the client of the result via an \MName{Inform} message. A client $\Client{}$ randomly sends a transaction $\T$ to a replica $\Replica_j$, starts a timer $t_C$, and awaits for $\f+1$ identical \MName{Inform} responses. If $\Client{}$ fails to get these $\f+1$ responses before $t_C$ expires, it sends the transaction to the next replica $\Replica_{j+1}$ and doubles the timeout interval. The client keeps doing so until it receive $\f+1$ identical \MName{Inform} responses confirming that the transaction has been executed. Due to primary rotation, \emph{every} replica will eventually be the primary of an instance that can propose $\T$, including a non-faulty replicas that will eventually propose $\T$.

Note that transaction execution in view $\v$ requires that all instances finished view $\v$ (hence, we know the total ordering of the transactions successfully proposed in view $\v$ and that need to be executed). In cases where system load is low, an instance primary $\Primary{}$ can end up not receiving any client transactions while other primaries are already proposing transactions for their instances. To prevent that executing of the proposals of other instances has to wait until $\Primary{}$ is able to propose a transaction, primary $\Primary{}$ can propose a no-op transaction if no other transactions are available. As we already proved \emph{termination} (Theorem~\ref{thm:cterm}) and \emph{non-divergence} (Theorem~\ref{thm:safety}) and the mechanisms outlined above guarantee \emph{service}, we finally conclude

\begin{theorem}
\SpotLess{} provides consensus.
\end{theorem}

\section{Evaluation}\label{sec:eval}

Previously, we detailed and analyzed the design of \SpotLess{}, showing several theoretical advantages when compared to its peers. Next, to show the practical advantages of \SpotLess{}, we will experimentally evaluate its performance, both in the normal case and during Byzantine failures. In our evaluation, we compare the performance of \SpotLess{} in \Changed{\RDB{} (Incubating)}, our high-performance open-source blockchain database, with the well-known primary-backup consensus protocols \PBFT{}~\cite{pbftj}, \HS{}~\cite{hotstuff}, \Narwhal{}~\cite{narwhal}, and our \PBFT{}-based concurrent consensus paradigm \RCC{}~\cite{rcc}. We focus on answering the following questions:

\begin{enumerate}
\renewcommand{\labelenumi}{Q\arabic{enumi}}
\item\label{q:1} Scalability: does \SpotLess{} deliver on the promises to provide better scalability than other consensus protocols?
\item\label{q:2} Latency: does \SpotLess{} provide low client latency while providing high throughput? What factors affect latency and throughput?
\item\label{q:3} What is the impact of batching client transactions on the performance of \SpotLess{}?
\item\label{q:4} How does \SpotLess{} perform in presence of Byzantine failures?
\item\label{q:5} How does \emph{concurrent consensus} improve performance?
\end{enumerate}

To study the practical performance of \SpotLess{} and other consensus
protocols, we implemented \SpotLess{} and other protocols in \Changed{\RDB{}}. To generate experimental workloads, we used the \emph{Yahoo Cloud Serving Benchmark}~\cite{ycsb} provided by the
Blockbench macro benchmarks~\cite{blockbench}. In the generated workload, each
client transaction queries a YCSB table with half a million active records and
$90\%$ of the transactions write and modify records. Before the experiments,
each replica is initialized with an identical copy of the YCSB table. We perform
all experiments on Oracle Cloud, using up to 128 machines for replicas and 32
machines for clients. Each replica and client is deployed on
a \texttt{e3}-machine with a $16$-core AMD EPYC 7742 processor, running at
$\SI{3.4}{\giga\hertz}$, and with $\SI{32}{\giga\byte}$ memory.

\subsection{The \Changed{\RDB{}} Blockchain Database}\label{eval:rdb}

The architecture of \Changed{\RDB{}} is heavily \emph{multi-threaded} and \emph{pipelined}
and is optimized for maximizing throughput~\cite{icdcs,resilientdb-demo,tut-middleware19,tut-debs20,tut-vldb20}. To further maximize throughput and minimize the overhead of any consensus protocol, \Changed{\RDB{}} has built-in support
for \emph{batching} client transactions.  We typically group $\SI{100}{\txn\per\batch}$. In this case, the size of a proposal
is $\SI{5400}{\byte}$ and of a client reply (for 100 transactions) is
$\SI{1748}{\byte}$. The other messages exchanged between replicas during the
replication algorithm have a size of $\SI{432}{\byte}$.

To maximize performance of \PBFT{}-like protocols, \Changed{\RDB{}} supports \emph{out-of-order processing} of transactions (see Section~\ref{sec:concurrent}) in which primaries can
propose future transactions before current transactions are executed. The chained designs of \SpotLess{} and \HS{} cannot utilize out-of-order processing, however.

In \Changed{\RDB{}}, each replica maintains an immutable \emph{blockchain ledger} that holds
an ordered copy of all executed transactions. The ledger not only stores all transactions,
but also strong cryptographic proofs of their acceptance by a consensus protocol. This ledger can be used to provide \emph{strong data provenance}.

In our experiments, replicas not only perform consensus but also communicate with
clients and execute transactions. In this practical setting, performance is not fully determined
by the cost of consensus, but also by the cost of \emph{communicating with clients},
of sequential \emph{execution} of all committed transactions, of \emph{cryptography}, and of other steps
involved in processing messages and transactions. The \emph{sequential execution} bottleneck of \Changed{\RDB{}} in our deployments, the maximum speed by which \Changed{\RDB{}} can execute transactions without performing any other tasks, is $\SI{340}{\kilo\txn\per\sec}$.

To reduce the overhead of sending individual messages, \Changed{\RDB{}} uses \emph{message buffering}: \Changed{\RDB{}} collects messages intended for specific receiver in a buffer and sends all messages at-once when the buffer reaches a threshold value. For each consensus protocol we study, the threshold value used has been configured to maximize throughput.

In addition, to better utilize the computing and network resources, \Changed{\RDB{}} optimizes \SpotLess{} in the following two ways. First, primaries broadcast the actual content of client requests in advance, and during the proposal step, they only operate on the digests of the requests. Second, in the geo-scale experiments that will be discussed later, primaries can broadcast a new proposal optimistically in a fast path before receiving 2f+1 votes for the proposal in the previous view. If any Byzantine behavior is detected, replicas switch back to the slow original path the next time they become the rotated primary.

\subsection{The Consensus Protocols}\label{eval:cons}

We evaluate the performance of \SpotLess{} by comparing it with a representative selection of four efficient practical consensus protocols implemented in \Changed{\RDB{}}:

\Paragraph{\PBFT{}~\cite{pbftj}.} We use heavily optimized out-of-order implementation that uses message authentication codes. 

\Paragraph{\RCC{}~\cite{rcc}.} RCC turns \PBFT{} into a concurrent consensus protocol.

\Paragraph{\HS{}~\cite{hotstuff}.} \HS{} uses threshold signatures to minimize communication. Since existing threshold signature algorithms are expensive and quickly become the bottleneck, we use a list of $\n-\f$ \texttt{secp256k1} digital signatures to represent a threshold signature, which improves our throughput. In our experiments, we implement the pipelined \Chained{} \HS{}.

\Paragraph{\Narwhal{}~\cite{narwhal}.} \Narwhal{} separates the replication of transactions and ordering transactions, enabling concurrent transaction dissemination. We simulate the communication complexity and computational overhead of \Narwhal{} by running \HS{} and requring replicas to broadcast messages consisting of a client batch and $2\f+1$ digital signatures.

We run the concurrent protocols \RCC{} and \SpotLess{} with $\n$ instances unless stated otherwise.

\subsection{Experiments}\label{eval:meas}

To be able to answer Questions Q\ref{q:1}--Q\ref{q:5}, we perform \Changed{fifteen} experiments
in which we measure the performance of \SpotLess{} and other consensus protocols. We measure
\emph{throughput} as the number of transactions that are executed per second and
\emph{latency} as the average duration between when the client sends a transaction and
when the client receives $\f+1$ corresponding responses. Unless stated otherwise, all replicas are non-faulty. We run each experiment for
$\SI{130}{\second}$ (except the seventh experiment): the first
$\SI{10}{\second}$ are warm-up, and measurement results are collected over the next
$\SI{120}{\second}$. We average our results over three runs.

In the \emph{scalability experiment}, we measure throughput as a function of the number of replicas. We vary
the number of replicas between $\n=4$ and $\n=128$ and we use a batch size of
$\SI{100}{\txn\per\batch}$. The results can be found in Figure~\ref{fig:spotless_plots}(a).

In the \emph{batching experiment}, we measure the throughput as a function of the number of replicas. We use $\n = 128$ replicas and we vary batch size between $\SI{10}{\txn\per\batch}$ and $\SI{400}{\txn\per\batch}$. The results can be found in Figure~\ref{fig:spotless_plots}(b). 

In the \emph{throughput-latency} experiment, we measure the latency as a function of the throughput. We use $\n=128$ replicas, set the batch size to $\SI{100}{\txn\per\batch}$, and we vary the speed by which each primary receives client requests to affect
throughput and latency. The results can be found in Figure~\ref{fig:spotless_plots}(c).

In \emph{the transaction-size} experiment, we measure the throughput of \SpotLess{} as a function of the  
individual YCSB transaction size. We use $\n=128$ replicas and vary the transaction size from $\SI{48}{\byte}$ to 
$\SI{1600}{\byte}$. The results can be found in Figure~\ref{fig:spotless_plots}(d).

\begin{figure}[t]
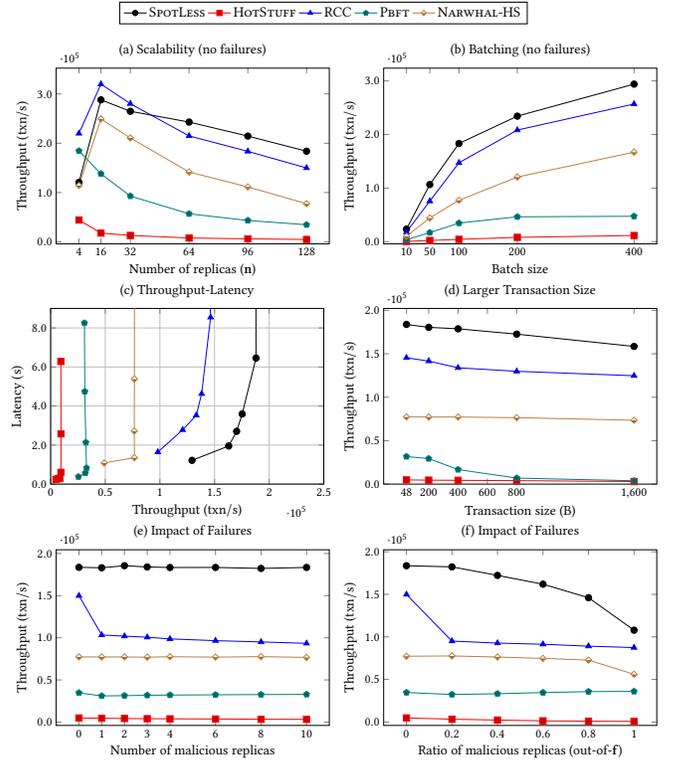

    \centering
    \scalebox{0.5}{\ref{mainlegend}}\\[5pt]
    \setlength{\tabcolsep}{1pt}
    \begin{tabular}{cc@{\quad}cc}
    \resultgraphd{\dataTputNodesFF}{(a) Scalability (no failures)}{\axisnodes}{\axistput}{\axisticksnodes}&
    \resultgraph{\dataTputBatch}{(b) Batching (no failures)}{\axisbatches}{\axistput}{\axisticksbatches}\\
    \resultgraphtputlat{\dataTputLatc}{(c) Throughput-Latency}{\axistput}{\axislat}{\axistickstput}&
    \resultgraphtxnsize{\dataTputTxnSize}{(d) Larger Transaction Size}{\axistxnsize}{\axistput}{\axistickstxnsize}\\
    \resultgraphfailureb{\dataTputCFMb}{(e) Impact of Failures}{\axisfailures}{\axistput}{\axisticksfailures}&
    \resultgraphfailurec{\dataTputPFMb}{(f) Impact of Failures}{\axispercentage}{\axistput}{\axistickspercentage}
    \end{tabular}
    \caption{Performance of \SpotLess{} and other protocols. \Changed{The number of replicas is 128 except in (a).}}\label{fig:spotless_plots}
\end{figure}

\begin{figure}
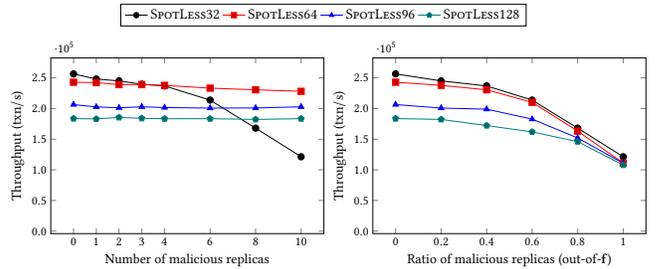

    \centering
    \scalebox{0.5}{\ref{mainlegend2}}\\[5pt]
    \setlength{\tabcolsep}{1pt}
    \begin{tabular}{cc@{\quad}cc}
    \resultgraphfailure{\dataTputCF}{}{\axisfailures}{\axistput}{\axisticksfailures}&
    \resultgraphfailure{\dataTputPF}{}{\axispercentage}{\axistput}{\axistickspercentage}
    \end{tabular}
    \caption{Performance of \SpotLess{} during failures as a function of the number of replicas and faulty replicas.}
\label{fig:spotless_failure}
\end{figure}

In the \emph{all-throughput-failures} experiment, we measure the throughput as a function of the number of malicious replicas that do not participate in consensus.
We use $\n=128$ replicas, and we vary the number of faulty replicas between $0$ and $10$ or between  $0$ and $\f$. We make the faulty replicas non-responsive at the same time point and measure throughput afterward for $\SI{120}{\second}$. We set the timeout length in \SpotLess{}, \HS{}, and \Narwhal{}
based on the calculated average view duration and set the timeout length in \RCC{} and \PBFT{} based on the
average client latency. The results can be found in Figure~\ref{fig:spotless_plots}(e) and (f).

In the \emph{$\SpotLess$-throughput-failures} experiment, we further measure the throughput of \SpotLess{} as a function of the number of replicas and the number of malicious replicas that do not participate in consensus. We vary the number of replicas between $\n \in \{ 32, 64, 96, 128 \}$ and we perform two measurements for each $\n$. We vary the number of
faulty replicas between $0$ and $10$ or between  $0$ and $\f$. Based on the calculated average view duration, we have set the timeout length in \SpotLess{} appropriately. The results can be found in Figure~\ref{fig:spotless_failure}.

In the \emph{throughput-latency-failure experiment}, we measure the latency of \SpotLess{} and \RCC{} as a 
function of throughput in the presence of malicious replicas. We use $\n=128$ replicas and set the 
number of faulty, non-responsive, replicas to be $1$ or $\f$. We vary the number of client 
batches that each primary receives in the same way as described in the throughput-latency experiment. We only count the latency of proposals that 
are sent to non-faulty replicas. The results can be found in Figure~\ref{fig:tput_lat_crash}.

In the \emph{parallel transaction processing experiment}, we measure the performance of \SpotLess{} and \RCC{} as a function of the amount of \emph{concurrent} (both protocols) and \emph{out-of-order} (\RCC) processing. To do so, we measure the throughput and latency of \SpotLess{} and \RCC{} as a function of the number of client batches that each primary receives. We use $\n = 128$ replicas and we vary the number of client batches between $12$ and $200$. The results can be found in Figure~\ref{fig:inflight}.

\begin{figure}[t]
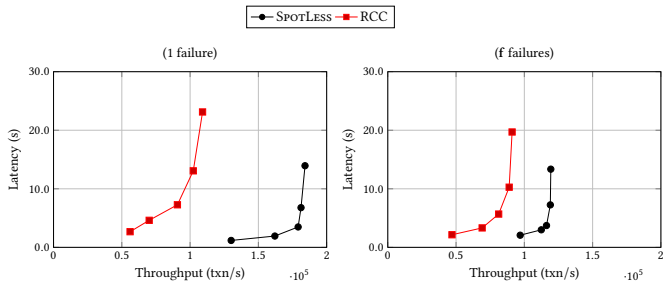

    \centering
    \scalebox{0.5}{\ref{mainlegend5}}\\[5pt]
    \setlength{\tabcolsep}{1pt}
    \begin{tabular}{cc@{\quad}cc}
    \resultgraphtputlatF{\dataTputLatFone}{(1 failure)}{\axistput}{\axislat}{\axistickstput}&
    \resultgraphtputlatF{\dataTputLatFF}{($\f$ failures)}{\axistput}{\axislat}{\axistickstput}
    \end{tabular}
    \caption{System throughput-latency of \SpotLess{} and \RCC{} in the presence of failures (128 replicas).}
    \label{fig:tput_lat_crash}
\end{figure}

In the \emph{throughput-Byzantine} experiment, we measure the throughput of \SpotLess{} in the presence of attacks as a function of the number of Byzantine replicas. We consider four types of attacks:

\Changed{
\begin{enumerate}
    \renewcommand{\labelenumi}{A\arabic{enumi}}
    \item faulty replicas are non-responsive;
    \item faulty replicas act \emph{malicious} when they are the primary by keeping $\f$ non-faulty replicas \emph{in the dark} (by not sending proposals to them);
    \item faulty replicas act \emph{malicious} by sending conflicting concurring votes in an attempt to cause \emph{divergence}: they send one message to $\f$ non-faulty replicas and a different one to the other non-faulty replicas; and
    \item faulty replicas act \emph{malicious} by refusing to participate in the consensus of proposals from non-faulty primaries, this in an attempt to subvert non-faulty primaries (and make them look faulty).
\end{enumerate}
}

\begin{figure}[t]
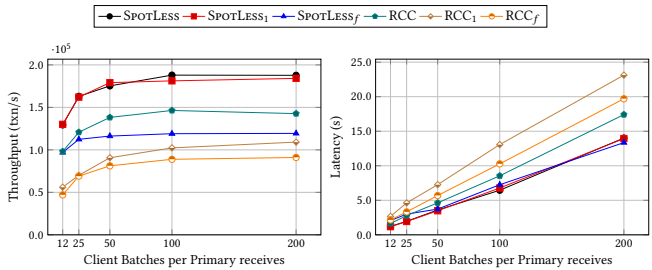

    \centering
    \scalebox{0.5}{\ref{mainlegend9}}\\[5pt]
    \setlength{\tabcolsep}{1pt}
    \begin{tabular}{cc@{\quad}cc}
    \flighttput{\dataFlightTput}{}{\axisinflight}{\axistput}{\axisticksinflight}&
    \flighttput{\dataFlightLat}{}{\axisinflight}{\axislat}{\axisticksinflight}
    \end{tabular}
    \caption{Throughput and latency of \SpotLess{} and \RCC{} as a function of the amount of parallel transaction processing (128 replicas).}\label{fig:inflight}
\end{figure}

For comparison, we include \RCC{}. The \emph{victims} of these attacks are the replicas that are kept in the dark (A2), receive a proposal 
that is received by not more than $\f$ non-faulty replicas (A3), or are not responded to (A4). The throughput of \RCC{} is not influenced by A2, A3, and A4 as long as the number of \emph{victims} is not greater than $\f$. Hence, we only include the normal-case (failure-free) throughput of \RCC{} and the throughput during A1 for comparison. We use $\n = 128$ replicas and we vary the number of malicious replicas between $0$ and $10$ or between $0$ and $\f$. The results can be found in Figure~\ref{fig:byzantine_failure}.

In the \emph{real-time-throughput-failure} experiment, we measure the real-time throughput of
\SpotLess{} and \RCC{} after making malicious replicas non-responsive as a function of time. We use $\n=128$ replicas, and 
we set the number of faulty replicas to $1$ or $\f$. We run the experiments for $\SI{140}{\second}$, 
record throughput every 5 seconds, and have the failures happen at the $10$th second. The results 
can be found in Figure~\ref{fig:timeline}.

In the \emph{concurrent-consensus} experiment, we measure the throughput of \SpotLess{} and \RCC{} as a function of the number of replicas and concurrent instances. We use $\n \in \{ 32, 128 \}$ replicas and we vary the number of concurrent instances between $1$ and $\n$. The results can be found in Figure~\ref{fig:spotless_concurrent}.

\Changed{
In the \emph{computing-power-impact} experiment, we measure the throughput of \SpotLess{} and other protocols as a function of the number of CPU cores in each replica. We use $\n = 128$ replicas and we vary the number of CPU cores between 4 and 32. The results can be found in Figure~\ref{fig:resource_impact}(a).

In the \emph{network-bandwidth-impact} experiment, we measure the throughput of \SpotLess{} and other protocols as a function of the bandwidth. We use $\n = 128$ replicas and vary the bandwidth between $\SI{500}{\mega\bit\per\second}$ and $\SI{4000}{\mega\bit\per\second}$ using FireQOS~\cite{fireqos}, a program that helps configure traffic shaping on Linux. The results can be found in Figure~\ref{fig:resource_impact}(b).

In the \emph{global-regions} experiment, we measure the throughput of \SpotLess{} and other protocols as a function of the number of regions. We use $\n = 128$ replicas and vary the number of regions between 1 and 4. For each run, the 128 replicas are uniformly distributed in the regions Oregon, North Virginia, London, and Zurich. The results can be found in Figure~\ref{fig:resource_impact}(c) and (d).

In the \emph{non-concurrent-failure} experiment, we measure the throughput of single-instance \SpotLess{} and \HS{} as a function of the number of faulty replicas. We use $\n = 128$ replicas and vary the number of malicious replicas between $0$ and $\f$. The results can be found in Figure~\ref{fig:single_instance_failure}.
}



\begin{figure}
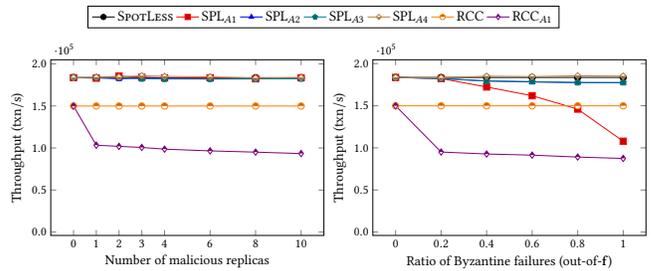

    \centering
    \scalebox{0.5}{\ref{byzantinelengend}}\\[5pt]
    \setlength{\tabcolsep}{1pt}
    \begin{tabular}{cc@{\quad}cc}
    \resultgraphfailurek{\dataTputBFa}{}{\axisbfailures}{\axistput}{\axisticksfailures}&
    \resultgraphfailurel{\dataTputBFb}{}{\axisbpercentage}{\axistput}{\axistickspercentage}
    \end{tabular}
    \caption{Performance of \SpotLess{} and \RCC{} during Byzantine failures (128 replicas, four attack scenarios for \SpotLess{}). \textit{Note: $\Name{SPL}$ is the abbreviation for \SpotLess{} here.}}\label{fig:byzantine_failure}
\end{figure}

\subsection{Experiment Analysis}\label{ss:discuss}

The experimental results presented in the previous sections allows us to answer the research questions~Q\ref{q:1}--Q\ref{q:5}. First, we observe that increasing
the batch size,  increases the performance of all consensus protocols, thereby answering Q\ref{q:3} as 
expected. Since the gains brought by increased the batch size are small after $\SI{100}{\txn\per\batch}$, we used $\SI{100}{\txn\per\batch}$ in all other experiments.

\SpotLess{} outperforms all other protocols in failure-free conditions (Q\ref{q:1}, Q\ref{q:2}). As Figure~\ref{fig:summary} shows, the amortized message complexity per decision is $\n^2$ for \SpotLess{}, while it is $2\n^2$ for \RCC{}. Hence, as \SpotLess{} has fewer messages to process, \SpotLess{} can even outperform \RCC{} by up to 23\%. Due to the \emph{message buffer mechanism}, \SpotLess{} and \RCC{} require sufficient batches of client requests to fill the system pipeline. Otherwise, the two protocols may get stalled since no messages are sent and processed. When the pipeline is full, the latency of both \SpotLess{} and \RCC{} in \Changed{\RDB{}} is dominated by the maximum throughput, as Figure~\ref{fig:spotless_plots}(c), \ref{fig:tput_lat_crash}, and \ref{fig:inflight} show. The higher throughput is, the shorter a client request waits to be proposed and then the lower latency is. Thus, even though \SpotLess{} needs more communication phases than \RCC{} to commit a proposal, \SpotLess{} has a lower latency by up to $32\%$ than \RCC{}. Also, \SpotLess{} outperforms \Narwhal{} because, for each committed block, \SpotLess{} verifies $O(n)$ MACs while \Narwhal{} verifies $O(n)$ digital signatures.

By introducing concurrent processing, \SpotLess{} is able to outperform \HS{}, the other chained consensus protocol, by up to 3803\%. In this situation, the performance of \HS{} is bottlenecked by the message delay due to the lack of out-of-order processing. From Figure~\ref{fig:spotless_plots}(d) we conclude that  \RCC{} and \SpotLess{} are able to sustain high throughput even if we increase  the transaction size to $\SI{1600}{\byte}$ per YCSB transaction, whereas the throughput of \PBFT{} and \HS{} decreases greatly. This is easily explained: the concurrent design of \RCC{} and \SpotLess{} load-balances the primary task to all replicas, whereas in \PBFT{} and \HS{} the performance is bottlenecked by the bandwidth available to the single proposing primary.

As Figure~\ref{fig:spotless_failure} shows, non-responsive faulty replicas negatively affect the 
performance of \SpotLess{} in all cases (Q\ref{q:4}): indeed, non-responsive faulty
replicas do not perform their primary role, due to which the non-faulty replicas can only wait until their timers expire to switch out these faulty primaries.
The larger the number of replicas, the smaller the relative influence of faulty replicas on performance.
For example, when there are $\f$ faulty replicas, the throughput of \SpotLess{128} decreases by 
41\% while that of \SpotLess{32} decreases by 54\%. Due to \emph{concurrent consensus}, a larger number of replicas implies more \SpotLess{} instances with non-faulty primaries that utilize CPU resources while waiting for the instances with non-responsive primaries.

From Figure~\ref{fig:spotless_plots}, \ref{fig:tput_lat_crash}, and \ref{fig:inflight}, we also know that \SpotLess{} shows great resilience 
to non-responsive faulty replicas when compared with other protocols (Q\ref{q:4}). In a deployment with $128$ replicas,
the first $\SI{120}{\second}$ after failures happen, \SpotLess{} shows a gain in throughput over other protocols and a lower latency (Q\ref{q:1}, Q\ref{q:2}) despite the number of faulty replicas.

\begin{figure}[t!]
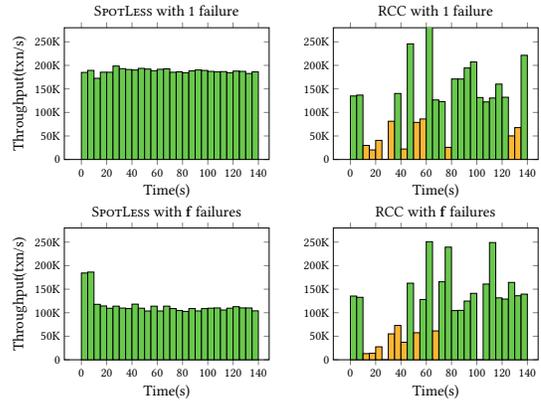

    \centering
    \setlength{\tabcolsep}{1pt}
    \begin{tabular}{cc@{\quad}cc}
    \graphSpotLess&
    \graphRCC&\\
    \graphSpotLessf&
    \graphRCCf
    \end{tabular}
    \caption{Throughput timeline of \SpotLess{} and \RCC{} after injecting failures.}\label{fig:timeline}
\end{figure}

\begin{figure}
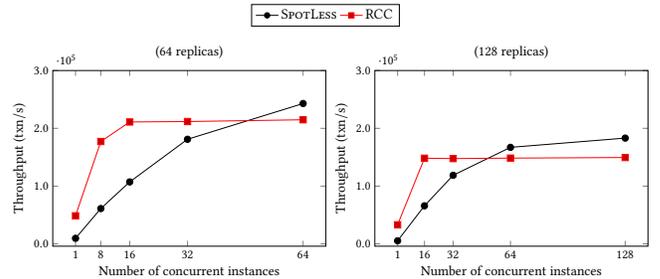

    \centering
    \scalebox{0.5}{\ref{mainlegend3}}\\[5pt]
    \setlength{\tabcolsep}{1pt}
    \begin{tabular}{cc@{\quad}cc}
    \resultgraphconcurrent{\dataTputCIa}{(64 replicas)}{\axisinstances}{\axistput}{\axisticksinstances}&
    \resultgraphconcurrent{\dataTputCIb}{(128 replicas)}{\axisinstances}{\axistput}{\axisticksinstancesb}
    \end{tabular}
    \caption{Performance of \SpotLess{} and \RCC{} as a function of the number of concurrent instances.}\label{fig:spotless_concurrent}
\end{figure}

Thanks to the \emph{\MName{Ask}-recovery mechanism} and \emph{Rapid View Synchronization}, described in Section~\ref{sec:design}, 
\SpotLess{} shows strong resilience to Byzantine attacks, as  we can observe from the results in Figure~\ref{fig:byzantine_failure}. No matter the type of attack,  the \emph{victims} can quickly detect the failure and catch up by receiving $\f+1$ \MName{Sync} 
messages from other non-faulty replicas and sending \MName{Ask} messages. When facing non-responsive 
faulty replicas, the two mechanisms are useless, however, as timing out instances is the only way to advance view in this case.

The results of Figure~\ref{fig:timeline} show obvious fluctuations in the real-time throughput of \RCC{} after injecting failures. This is due to the usage of an exponential back-off penalty algorithm to ignore instances with faulty primaries. Eventually, the throughput of \RCC{} gradually recovers to the original level and then keeps stable. This is the best case for \RCC{}, however, as all $\f$ failures happen at the same time. If the failures appear one by one, \RCC{} will suffer from these low-throughput  fluctuations (the yellow columns in Figure~\ref{fig:timeline})  during each failure. In contrast, \SpotLess{} presents a more stable throughput timeline after failures happen (Q\ref{q:4}).

Figure~\ref{fig:spotless_concurrent} shows that \SpotLess{} benefits more from concurrent consensus than \RCC{}, especially when there are many instances. (Q\ref{q:5}). When there are 32 or fewer instances, \RCC{} outperforms \SpotLess{} because \RCC{} enables \emph{out-of-order processing} in individual instances, which is not supported by the chained design of \SpotLess{} instances. As the number of concurrent instances increases, the throughput of \RCC{} reaches a message processing bottleneck when there are 16 instances and then remains stable, whereas the throughput of \SpotLess{} can increase further due to the lower message complexity and reaches its peak value when there are $\n$ instances, higher than \RCC{} by up to 23\%.

\Changed{The performance of consensus protocols is significantly influenced by computing and network resources.
First, Figure~\ref{fig:resource_impact}(a) shows that the performance of all protocols decreases when compute power is restrictd (fewer CPU cores in each replica). Second, Figure~\ref{fig:resource_impact}(b) shows that decreasing the network bandwidth negatively impacts the performance of all protocols in which network is a bottleneck. We note that \Narwhal{} is barely affected, however, as it is limited by computing resources as it has to verify $\n-\f$ digital signatures per block. Similarly, Figure~\ref{fig:resource_impact}(c) and (d), show that \emph{increasing} the number of regions, which not only decreases network bandwidth but also increases latency, \emph{negatively impacts} the performance of all protocols. We notice that in all cases, \SpotLess{} maintains a higher performance than \RCC{}, even thoug \SpotLess{} is affected by decreasing bandwidth. Finally, the comparison between Figure~\ref{fig:resource_impact}(c) and (d), shows that that increasing the batch size can partially mitigates bandwidth bottlenecks.

Figure~\ref{fig:single_instance_failure} shows that the presence of Byzantine failures has similar negative effects on the performance of single-instance \SpotLess{} and \HS{}. In all cases, the throughput of single-instance \SpotLess{} is higher than that of \HS{} due to the lower computation costs of verifying signatures in \SpotLess{} (as compared to dealing with the threshold signatures used in \HS{}). Hence, compared with \HS{}, replicas in \SpotLess{} are able to respond more quickly, lowering the per-round latency and increasing throughput.}

Based on our findings, we conclude that \SpotLess{} makes full use of \emph{concurrent consensus} (Q\ref{q:5}),
provides higher throughput and better scalability than any other consensus protocol (Q\ref{q:1}), and does so with low latency (Q\ref{q:2}). Moreover, client batching
does benefit the performance of \SpotLess{} (Q\ref{q:3}). Finally, we conclude that \SpotLess{} can 
efficiently deal with failures (Q\ref{q:4}) thanks to \emph{concurrent consensus},  the
\emph{\MName{Ask}-recovery mechanism}, and \emph{Rapid View Synchronization}.

\begin{figure}[t]
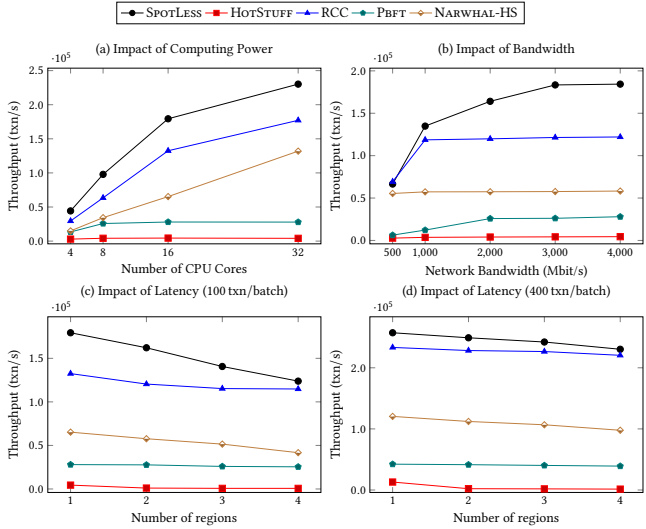

    \centering
    \scalebox{0.5}{\ref{mainlegend}}\\[5pt]
    \setlength{\tabcolsep}{1pt}
    \begin{tabular}{cc@{\quad}cc}
    \resultgraphresource{\dataTputCore}{\Changed{(a) Impact of Computing Power}}{\axiscore}{\axistput}{\axistickscores}{}&\resultgraphresource{\dataTputBandwidth}{\Changed{(b) Impact of Bandwidth}}{\axisbandwidth}{\axistput}{\axisticksbandwidth}{}\\
    \resultgraphresource{\dataTputRegion}{{\Changed{(c) Impact of Latency ($\SI{100}{\txn\per\batch}$)}}}{\axisregions}{\axistput}{\axisticksregions}{title style={yshift=1.5ex}}&
    \resultgraphresource{\dataTputRegionFour}{\Changed{(d) Impact of Latency ($\SI{400}{\txn\per\batch}$)}}{\axisregions}{\axistput}{\axisticksregions}{title style={yshift=1.5ex}}
    \end{tabular}
    \caption{\Changed{Impact of computing and network resources on the performance with 128 replicas.}}\label{fig:resource_impact}
\end{figure}

\begin{figure}
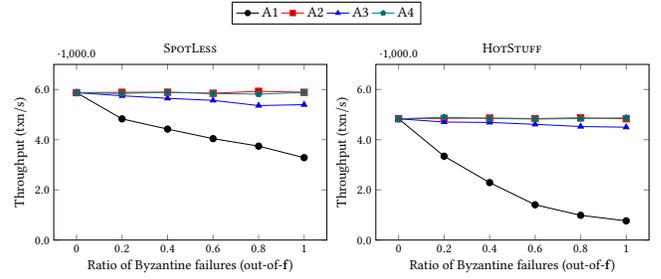

    \centering
    \scalebox{0.5}{\ref{si_legend}}\\[5pt]
    \setlength{\tabcolsep}{1pt}
    \begin{tabular}{cc@{\quad}cc}
    \resultgraphsifaiurer{\dataTputSingleInsatncePVP}{\SpotLess{}}{\axisbpercentage}{\axistput}{\axistickspercentage}&
    \resultgraphsifaiurer{\dataTputSingleInsatnceHS}{\HS{}}{\axisbpercentage}{\axistput}{\axistickspercentage}
    \end{tabular}
    \caption{\Changed{Performance of single-instance \SpotLess{} and \HS{} with failures.}}\label{fig:single_instance_failure}
\end{figure}

\section{Related Work}\label{sec:related}

There is abundant literature on consensus and primary-backup consensus in specific 
(e.g.,~\cite{wild,scaling,untangle,leaderless-consensus,next700bft,zyzzyvaj,sbft,ardagna2020blockchain,xft,
loghin2022blockchain,blockchain-info-share,lineagechain,bft-forensics,mirbft,gem2tree}), to reduce the 
communication cost and improve performance and resilience of the consensus systems~\cite{thetabyz,lewis,kogias2016enhancing,borealis,
shadoweth,occlum,eactors,experimental-bft-improv,bft-to-cft,kuhring2021streamchain}. 
In previous sections, we already discussed how \SpotLess{} relates to \PBFT{}~\cite{pbftj}, \RCC{}~\cite{rcc}, 
\HS{}~\cite{hs}, and \Narwhal{}~\cite{narwhal}. Next, we shall focus on 
other works that deal with either \emph{improving throughput and scalability} or with \emph{simplifying consensus}, 
the two strengths of \SpotLess{}.

\Paragraph{Trusted Hardware}
There is a large body of work on improving and simplifying primary-backup consensus by employing 
\emph{trusted hardware}. A representative example is \MinBFT{}~\cite{minbft}, which uses trusted hardware to 
prevent malicious primaries from proposing conflicting client requests in a round. Via the usage of trusted
hardware, one can reduce the impact of Byzantine behavior and simplify both the normal case and the recovery process, e.g., \MinBFT{} eliminates one round of communication from \PBFT{} and can operate in deployments in which $\n > 2\f$ (instead of the typical $\n > 3\f$) holds. Protocols that require trusted hardware place significant restrictions on the environment in which they can run, due to which their practical usage is limited. As such, these protocols are not a replacement for a general-purpose consensus such as \SpotLess{}.~\cite{disth}

\Paragraph{Leader-Less Consensus.} Leader-less protocols such as \Name{HoneyBadger}~\cite{badger} and \Name{Dumbo}~\cite{dumbo} eliminate the limitations of \PBFT{} and other primary-backup consensus protocols via a fully decentralized and fully asynchronous design. In these protocols, all replicas have the same responsibilities, due to which the cost of consensus is equally spread-out over all replicas. These leader-less protocols claim to improve resilience over \PBFT{} in asynchronous environments. Due to the high complexity of fully asynchronous consensus, their practical performance is limited, however.

\Paragraph{Sharding.}
Recently, there have been several approaches toward scalability of \RDMS{}s by sharding 
them, e.g.,~\cite{ahl,jbyshard,byshard,sharper,blockchaindb}. Although sharding has the potential to drastically 
improve scalability for certain workloads, it does so at a high cost for complex workloads. Furthermore, 
sharding impacts resilience, as sharded systems put requirements on the number of failures in each individual 
shard (instead of putting those requirements on all replicas in the system). Finally, sharding is orthogonal to 
consensus, as proposed sharded systems all require a high-performance consensus protocol to run individual 
shards. For this task, \SpotLess{} is an excellent candidate.

\Paragraph{Reducing Primary Costs.}
There are several approaches toward reducing the cost for the primary to coordinate consensus in 
\PBFT{}-style primary-backup consensus protocols, thereby reducing the limitations of primary-backup 
designs. Examples include (1) protocols such as \Name{FastBFT}~\cite{fastbft} and the geo-scale aware \Name{GeoBFT}~\cite{geobft} that use a \emph{hierarchical communication architecture} to offload the cost for a primary to propose client requests to other replicas; (2) protocols such as \Narwhal{}~\cite{narwhal} that use a gossip-based communication protocol to replicate client requests, thereby sharply reducing the cost for primaries to propose these requests; and (3) protocols such as \Name{Algorand}~\cite{algorand} that restrict consensus to a small subset of the replicas in the system  (whom then enforce their decisions upon all other replicas), thereby reducing the cost of the primary to coordinate consensus. Unfortunately, each of these protocols introduces its own added complexity or environmental restrictions to the design of consensus, e.g., \Name{FastBFT} requires trusted hardware, the design of \Name{GeoBFT} impacts resilience in a similar way as sharded designs do, and \Name{Algorand} relies on complex cryptographic primitives due to which it can only guarantee to work with high probability.
\section{Conclusion}\label{sec:conclusion}
In this paper, we proposed \SpotLess{}, a high-performance robust consensus protocol. \SpotLess{} combines the high 
throughput of \emph{concurrent consensus architectures} with the reduced complexity provided by \emph{chained consensus}. 
Furthermore, \SpotLess{} improves on the resilience of existing chained consensus designs by introducing the 
\emph{Rapid View Synchronization protocol}, which guarantees a continuous low-cost recovery path that is robust 
during unreliable communication and does not require costly threshold signatures.

We have put the design of \SpotLess{} to the test by implementing it in \RDB{}, our high-performance resilient fabric, 
and we compared \SpotLess{} with existing consensus protocols. Our experiment results show that the performance of \SpotLess{} is 
excellent: \SpotLess{} greatly outperforms traditional primary-backup consensus protocols such as \PBFT{} by up to 430\%, 
\Narwhal{} by up to 137\%,
and \HS{} by up to 3803\%. \SpotLess{} is even able to outperform \RCC{}, a state-of-the-art high-throughput concurrent 
consensus protocol, by up to 23\% in optimal conditions, while providing lower latency in all cases. Furthermore, 
\SpotLess{} can maintain a stable latency and consistently high throughput, even during failures.

\begin{acks}
This work was supported in part by (1) {\em Oracle Cloud Credits} and 
related resources provided by the Oracle for Research program and (2) the 
{\em NSF STTR} under {\em Award Number} 2112345 provided to {\em Moka Blox LLC}.
\end{acks}

\bibliographystyle{ACM-Reference-Format}
\bibliography{sources}

\end{document}